\documentclass[12pt]{article}
\usepackage{amsmath, amssymb, amsthm}
\usepackage[margin=1in]{geometry}
\usepackage{natbib}
\usepackage{color}
\usepackage[hidelinks]{hyperref}
\usepackage{graphicx}
\usepackage{booktabs}
\usepackage{multirow}
\usepackage{subcaption}
\newtheorem{proposition}{Proposition}

\newcommand{\RAH}{\mathrm{RAH}}

\newcommand{\DRMST}{D\text{-}\mathrm{RMST}}
\newcommand{\dd}{\,d}
\newcommand{\TauSim}{40}
\newcommand{\TauCM}{24}
\newcommand{\CMAdminLower}{13}
\newcommand{\CMAdminUpper}{28}
\newcommand{\CMDropoutRate}{0.005}
\newcommand{\CMDropoutAnnualPct}{6}
\newcommand{\TBFading}{43}
\newcommand{\ScTauLight}{0.154}
\newcommand{\ScTauHeavy}{0.058}

\newcommand{\AreRahCoxCM}{1.067}
\newcommand{\AreRahDrmstCM}{1.372}

\newcommand{\CosSqCoxCM}{0.603}
\newcommand{\CosSqDrmstCM}{0.469}
\newcommand{\CosSqRahCM}{0.643}

\newcommand{\PowerCoxCMPct}{85}
\newcommand{\PowerDrmstCMPct}{72}
\newcommand{\PowerRahCMPct}{89}
\newcommand{\NPowerCoxCM}{300}
\newcommand{\NPowerDrmstCM}{400}
\newcommand{\NPowerRahCM}{250}
\newcommand{\NPowerCoxMorePct}{20}
\newcommand{\NPowerDrmstMorePct}{60}
\newcommand{\SchoenfeldPCM}{0.002}
\newcommand{\WeibullShapeControl}{1.264}
\newcommand{\WeibullScaleControl}{13.076}
\newcommand{\WeibullShapeTreatment}{0.987}
\newcommand{\WeibullScaleTreatment}{17.320}

\newcommand{\SzeroCMTau}{0.12}
\newcommand{\AreRahCoxPHHeavyLowS}{0.798}
\newcommand{\AreRahDrmstPHHeavyLowS}{0.812}
\newcommand{\AreRahCoxPHShapeHalf}{0.922}
\newcommand{\AreRahCoxPHShapeTwo}{0.737}
\newcommand{\AreDrmstCMTauEighteen}{0.534}

\newcommand{\AreDrmstCMTauTwentyfour}{0.778}
\newcommand{\AreRahCMTauEighteen}{1.063}
\newcommand{\AreRahCMTauTwentyone}{1.080}
\newcommand{\AreRahCMTauTwentyfour}{1.067}
\newcommand{\AreRahDrmstCMTauEighteen}{1.990}
\newcommand{\AreRahDrmstCMTauTwentyfour}{1.372}
\newcommand{\AreDrmstCMDropoutLow}{0.749}
\newcommand{\AreDrmstCMDropoutHigh}{0.891}
\newcommand{\AreRahCMDropoutLow}{1.047}
\newcommand{\AreRahCMDropoutHigh}{1.136}
\newcommand{\AreRahDrmstCMDropoutLow}{1.398}
\newcommand{\AreRahDrmstCMDropoutHigh}{1.274}
\newcommand{\AreDrmstCMPH}{0.988}
\newcommand{\AreRahCMPH}{0.909}

\title{A Cosine-Projection Framework for Choosing Among\\
       Cox HR, D-RMST, and Average-Hazard Tests in Survival Trials}
\author{Xiang Meng, Lu Tian, Hajime Uno}
\date{October 10, 2026}

\begin{document}
\maketitle

\begin{abstract}
The log-rank/Cox hazard-ratio test remains the standard approach for the primary analysis of randomized trials with time-to-event endpoints and is fully efficient under proportional hazards (PH). Two model-free alternatives that have recently gained attention are the difference in restricted mean survival time (D-RMST), a contrast on the time scale, and the more recently proposed ratio of average hazards (RAH), a contrast on the rate scale. Simulation studies suggest that the RAH-based test can outperform both under delayed treatment effects, but its efficiency has not been characterized analytically, and no unified framework explains how censoring and accrual affect the efficiency ordering among these tests.
We show that all three statistics share a weighted-log-rank representation and derive their asymptotic relative efficiencies under local alternatives, including the first analytic efficiency results for the RAH-based test.
Geometrically, each test defines a direction in a function space and the anticipated treatment effect defines a signal direction. A test's efficiency relative to the oracle test equals the squared cosine of the angle between these directions, so the asymptotic relative efficiency of any two tests is a ratio of squared cosines. Because censoring and accrual enter through an information weight, the framework shows how design choices can alter, and even reverse, the efficiency ordering.
We illustrate the framework using the CheckMate~057 immunotherapy trial and corroborate the theoretical results through simulations spanning proportional- and non-proportional-hazards settings.
\end{abstract}

\section{Introduction}
\label{sec:intro}

Non-proportional hazards arise in randomized clinical trials with time-to-event endpoints, particularly in immuno-oncology, where treatment effects may be delayed and hazards may cross \citep{Horiguchi2019,Uno2023}. In such settings, the choice of the primary analysis method is an important design decision.
This paper compares three tests that a trialist might reasonably consider.
The first is the conventional test based on the Cox proportional hazards model, asymptotically equivalent to the log-rank test. For the primary analysis of overall survival, the 2025 FDA draft guidance recommends the log-rank test when there is alpha allocation, accompanied by the hazard ratio (HR) from a Cox model as the summary measure \citep{FDA2025OS}, and the test is asymptotically fully efficient when the proportional hazards (PH) assumption holds.
The second is the test based on the difference in restricted mean survival time (D-RMST) up to a prespecified time horizon, an established alternative that compares survival on the time scale without requiring the PH assumption \citep{Royston2013,Uno2014,Tian2014,survRM2}.
The third is the test based on the ratio of average hazards (RAH), proposed recently by \citet{Uno2023}. Average hazard summarizes the population-level event rate per unit of event-free person-time over a prespecified time horizon, and the RAH compares these rates between groups. Like D-RMST, RAH is defined without requiring the PH assumption, but it summarizes the between-group contrast on the rate scale rather than the time scale \citep{Uno2023,meng2026efficient}.

The choice among these tests can have substantial consequences at the design stage. Power calculations for time-to-event endpoints require specification of the baseline survival distribution, the anticipated treatment-effect pattern, the censoring and accrual distributions, and the analysis time horizon. The relative efficiency of the Cox-, D-RMST-, and RAH-based tests therefore depends on these design features. This motivates the practical question: \emph{given these design specifications, which of the three tests is the most efficient?}

Previous work has examined this question from both theoretical and empirical perspectives. On the theoretical side, \citet{Tian2018} compared the log-rank and D-RMST tests through their asymptotic relative efficiency (ARE) under local alternatives. They showed that the D-RMST test is nearly as efficient as the log-rank test under PH, can be more efficient under some non-PH alternatives, but can lose efficiency when the treatment effect emerges primarily later in follow-up.
On the empirical side, \citet{Uno2023} compared average-hazard-based tests with the Cox- and RMST-based tests through simulation under PH, early-difference, and delayed-difference patterns and under different censoring distributions.
Their results showed that the relative performance of the tests can vary substantially with both the treatment-effect pattern and the censoring distribution.
Despite these developments, the RAH-based test has not been placed on the same theoretical efficiency scale as the Cox- and D-RMST-based tests. Its power properties have been documented by simulation, but no analytic expression for its efficiency is available, and no unified framework explains when and why the efficiency ordering among the three tests arises or how censoring and accrual can alter that ordering.

We develop such a unified theoretical framework. We first show that all three statistics admit a common weighted-log-rank representation and derive their asymptotic relative efficiencies under local alternatives, including the first analytic efficiency results for the RAH-based test.
Although the resulting weight functions indicate where each test places emphasis over follow-up, their shapes alone do not determine relative efficiency because the amount of information available at each time also depends on the trial design. We therefore recast the comparison geometrically. Each test corresponds to a direction in a function space determined jointly by its weight function and the design's information weight, while the treatment-effect pattern defines a signal direction in the same space. The efficiency of a test relative to the oracle test for that alternative is exactly the squared cosine of the angle between these two directions, and the ARE of any two tests is therefore a ratio of squared cosines. This \emph{cosine metric} turns efficiency comparisons into statements about alignment and makes explicit how censoring and accrual can alter, and even reverse, the efficiency ordering among tests.

We evaluate the theoretical results through simulations under PH and non-PH scenarios, including fading and delayed treatment effects, and illustrate the design-stage use of the framework in a case study based on reconstructed CheckMate~057 data. These analyses assess how well the theoretical efficiency comparisons predict empirical performance across different treatment-effect patterns and trial designs.

The remainder of the paper is organized as follows. Section~\ref{sec:setup} defines the three test statistics, derives their common weighted-log-rank representation, and presents their asymptotic efficiencies under local alternatives. Section~\ref{sec:cosine} develops the cosine metric and its geometric interpretation. Section~\ref{sec:sim} reports the simulation studies, and Section~\ref{sec:cm057-analysis} presents the CheckMate~057 design analysis. Section~\ref{sec:discussion} concludes with practical implications and directions for future work.

\section{Three Test Statistics, a Weighted-Log-Rank Representation, and Their Efficiency}
\label{sec:setup}

\subsection{Setup and notation}
\label{sec:notation}

Let $T_k$ ($k=0,1$) denote the event time in group $k$ with survival
function $S_k(t)$, cumulative hazard function $\Lambda_k(t)$, and hazard function $\lambda_k(t)$.
Define
\begin{equation}
  F_k(t)=1-S_k(t)\qquad \mbox{and} \qquad R_k(t)=\int_0^t S_k(u)\dd u.
\end{equation}
Thus, $F_k(t)=\Pr(T_k\le t)$ is the cumulative incidence (the probability
that the event has occurred by time $t$) and $R_k(\tau)$ is the RMST on $(0,\tau]$.
Let $C_k$ be the censoring time, independent of $T_k$, with
$G_k(t)=\Pr(T_k\wedge C_k\ge t)$ bounded away from zero on $(0,\tau]$.
Let $p_k$ be the limiting proportion of subjects in group $k$ $\left(p_0+p_1=1\right)$.
Lastly, the at-risk weight function is defined as
\begin{equation}\label{eq:w}
  w(t)=\frac{p_0\,p_1\,G_0(t)\,G_1(t)}{p_0\,G_0(t)+p_1\,G_1(t)}.
\end{equation}

Two simplifying conditions recur throughout the paper.
By \emph{equal allocation} we mean a balanced design with $p_0=p_1=1/2.$
By \emph{a common censoring distribution} we mean that the two groups
share the same marginal censoring law, and we write
$S_C(v)=\Pr(C\ge v)$ for the common censoring survival function;
this is the case for administrative censoring.
Under these two conditions, the at-risk weight \eqref{eq:w} becomes $w(t)=0.5\times S_C(t)/\{S_0(t)^{-1}+S_1(t)^{-1}\}.$ If the null hypothesis, $S_1(t)=S_0(t), t\in [0, \tau],$ is also true, then  the weight function further reduces to $w(t)=S_0(t)S_C(t)/4$.

\subsection{The three test statistics}
\label{sec:three-stats}

We can test $H_0:S_1(t)=S_0(t), t\in [0, \tau]$ against alternatives using three
statistics.  For each statistic, we state the corresponding estimand, point estimator, and asymptotic distribution of the test statistic.
Throughout, $\hat S_k(\cdot)$ and
$\hat\Lambda_k(\cdot)$ denote the Kaplan--Meier and Nelson--Aalen estimators of
$S_k(\cdot)$ and $\Lambda_k(\cdot)$, respectively, $\hat G_k(\cdot)$ denote the empirical counterpart of $G_k(\cdot)$, and $\hat{R}_k(t)$ denote $\int_0^t \hat{S}_k(u)du$, a consistent estimator of $R_k(t).$

\paragraph{Cox's HR-based test.}
Under the PH assumption,
$$\theta=\log\left\{\frac{\lambda_1(t)}{\lambda_0(t)} \right\}$$
 does not depend on $t$ and can be consistently estimated by the maximum partial likelihood estimator $\hat{\theta}.$ As $n\to \infty,$
\[
  \sqrt n(\hat\theta-\theta)\xrightarrow{d}N(0,\sigma_\theta^2),
\]
where
\begin{equation}\label{eq:var-cox}
  \sigma_\theta^2
    =\left\{\int_0^{\tau}
       \frac{e^{\theta}p_0p_1\,G_0(u)\,G_1(u)}
            {p_0\,G_0(u)+e^{\theta}p_1\,G_1(u)}\,\dd\Lambda_0(u)
      \right\}^{-1}.
\end{equation}
Under $H_0:\lambda_1(\cdot)=\lambda_0(\cdot)$,
$$\sigma_\theta^2=\left\{\int_0^{\tau}w(u)\dd\Lambda_0(u)\right\}^{-1},$$
 and the
Wald statistic
\(
  Z_{\mathrm{Cox}}
    =\sqrt n\,\hat\theta/\hat\sigma_\theta
\)
is asymptotically standard normal, where $\hat{\sigma}_\theta^2$ is a consistent estimator of $\sigma_\theta^2.$

 We use $Z_{\rm Cox}$ for all numerical studies in this paper, whereas its theoretical properties were examined based on an asymptotically equivalent score statistic, i.e., the
log-rank statistic. The details are given in Section~\ref{sec:wlr} and Appendix~\ref{sec:appendix-weights}.

\paragraph{D-RMST-based test.}
The \textbf{D-RMST} test is based on the difference between two estimated RMSTs:
\begin{equation}\label{eq:rmst-stat}
  \hat\theta_{\DRMST}=\hat R_1(\tau)-\hat R_0(\tau)
            =\int_0^\tau\bigl[\hat S_1(t)-\hat S_0(t)\bigr]\dd t,
\end{equation}
the area between the two Kaplan--Meier curves up to $\tau$
\citep{Uno2014,survRM2}. Let $\theta_{\DRMST}=R_1(\tau)-R_0(\tau)$ be the true RMST difference. It is a clinically interpretable estimand, and oftentimes preferred when the PH assumption may be violated.  Then, $\hat\theta_{\DRMST}$ is a consistent estimator of $\theta_{\DRMST}$ and
$\sqrt{n}(\hat\theta_{\DRMST}-\theta_{\DRMST})
\xrightarrow{d}N(0,\sigma_{\DRMST}^2)$ with
\begin{equation*}\label{eq:var-rmst}
  \sigma_{\DRMST}^2
    =\sum_{k=0}^{1}
     \int_0^\tau\!\left\{\int_v^\tau S_k(u)\dd u\right\}^{\!2}
     \frac{\dd\Lambda_k(v)}{p_k\,G_k(v)},
\end{equation*}
as $n \to \infty$
\citep{Tian2018}. Under the null, the Wald statistic
\[
  Z_{\DRMST}
    =\sqrt n\,\hat{\theta}_{\DRMST}/\hat\sigma_{\DRMST}
\]
is asymptotically standard normal, where $\hat{\sigma}_{\DRMST}^2$ is a consistent estimator of $\sigma_{\DRMST}^2.$

\paragraph{RAH-based test.}
The AH for group $k$ on $(0,\tau]$ is
\begin{equation*}\label{eq:AH}
  \eta_k(\tau)=\frac{F_k(\tau)}{R_k(\tau)}
              =\int_0^\tau\lambda_k(t)\,
               \frac{S_k(t)}{\int_0^\tau S_k(u)\dd u}\dd t,
\end{equation*}
a weighted average of hazard function over $[0, \tau]$ \citep{Uno2023}. %
$\eta_k(\tau)$ is a population-level person-time event rate.
Two properties make
AH attractive as a robust summary on the hazard scale.  First, $\eta_k(\tau)$ is defined purely through the event-time distribution
and is therefore not affected by the censoring distribution. On the other hand, the commonly used person-time incidence rate $\Pr(T\le C)/\mathbb{E}(T\wedge C)$ depends on both survival and censoring distributions \citep{Uno2023}.
Second, AH respects stochastic ordering: if $S_1(t)\ge S_0(t)$ for all
$t\le\tau$, then $\eta_1(\tau)\le\eta_0(\tau)$.
The ratio of AH (RAH) between two groups is
\begin{equation}\label{eq:RAH}
  \theta_{\RAH}
    =\frac{\eta_1(\tau)}{\eta_0(\tau)}
    =\frac{F_1(\tau)/R_1(\tau)}{F_0(\tau)/R_0(\tau)},
\end{equation}
with $\theta_{\RAH}=1$ under $H_0$.  Its estimator $\hat\theta_{\RAH}$ can be obtained by
replacing $F_k(\tau)$ and $R_k(\tau)$ with $1-\hat{S}_k(\tau)$ and $\hat R_k(\tau)$ in \eqref{eq:RAH}, respectively. As $n \to \infty$, $\sqrt{n}(\log\hat\theta_{\RAH}-\log\theta_{\RAH})\xrightarrow{d}N(0,\sigma_{\RAH}^2)$, \citep{Uno2023},
where
\begin{equation}\label{eq:var-RAH}
  \sigma_{\RAH}^2
    =\sum_{k=0}^{1}
     \int_0^\tau
     \!\left\{\frac{1}{F_k(\tau)}-\frac{R_k(u)}{R_k(\tau)}\right\}^{\!2}
     \frac{\dd\Lambda_k(u)}{G_k(u)\,p_k}.
\end{equation}
The test statistic $Z_{\RAH}=\sqrt{n}\log(\hat\theta_{\RAH})/\hat\sigma_{\RAH}$
is asymptotically $N(0,1)$ under $H_0$, where $\hat{\sigma}_{\RAH}^2$ is a consistent estimator of $\sigma^2_{\RAH}.$

\subsection{Weighted-log-rank representation}
\label{sec:wlr}

All three statistics, to the first order, can be written as a weighted contrast of
Nelson--Aalen estimates of cumulative hazard functions from two groups:
\begin{equation}\label{eq:wlr-generic}
  U_W=\int_0^\tau W(t)d\hat{\Delta}_\Lambda(t),
\end{equation}
 where
\[
  \hat{\Delta}_\Lambda(v)=\hat\Lambda_1(v)-\hat\Lambda_0(v).
\]
Under $H_0,$ the two cumulative hazards coincide and $U_W$ has mean zero. The
weight $W(\cdot)$ decides which time periods contribute to the comparison, and
choosing $W(\cdot)$ to match the anticipated alternative is the classical route to
constructing powerful, distribution-free survival tests \citep{FlemingHarrington1991,Tian2018}. Rigorous arguments for expressing three aforementioned test statistics in this form can be found in Appendix~\ref{sec:appendix-weights}.

\paragraph{Cox's HR-based test.}
The log-rank statistic, the score statistic under Cox model, can be written as
\begin{equation}\label{eq:cox-stat}
  U_{\mathrm{Cox}}
    =\int_0^\tau \frac{p_1p_0\hat{G}_1(t)\hat{G}_0(t)}{p_1\hat{G}_1(t)+p_0\hat{G}_0(t)}
     d\hat{\Delta}_\Lambda(t).
\end{equation}
 Under
the local alternatives considered in section \ref{sec:efficiency}, a one-step expansion of the score function from the partial likelihood function around $\theta=0$ yields that
$\sqrt n\,\hat\theta=\sigma_\theta^2\sqrt n\,U_{\mathrm{Cox}}+o_p(1)$. Thus, the Wald statistic,  $Z_{\mathrm{Cox}}$ and the normalized $\sqrt{n}U_{\mathrm{Cox}}$ are asymptotically equivalent. Furthermore,
the random weight in \eqref{eq:cox-stat} converges to the at-risk weight $w(\cdot)$ in
\eqref{eq:w} in probability, as $n \to \infty.$

\paragraph{D-RMST- and RAH- based tests.}
Neither statistic is a log-rank type statistic by construction, but each is a
smooth functional of $\hat\Lambda_1(\cdot)$ and
$\hat\Lambda_0(\cdot)$.
For D-RMST, $R(\Lambda)=\int_0^\tau \exp\{-\Lambda(t)\}\dd t$, and under local alternative,
\begin{align*}
\hat{\theta}_{\DRMST}
=&\int_0^\tau S_0(t)\left[\exp\left\{-\hat{\Delta}_\Lambda(t)\right\}-1\right]\dd t+O\bigl(\|\hat{\Delta}_\Lambda\|^2\bigr)\\
=& -\int_0^\tau S_0(t)\,\hat{\Delta}_\Lambda(t)\dd t +O\bigl(\|\hat{\Delta}_\Lambda\|^2\bigr)\\
=&-\int_0^\tau \left\{R_0(\tau)-R_0(t)\right\}d\hat{\Delta}_\Lambda(t)+O\bigl(\|\hat{\Delta}_\Lambda\|^2\bigr) \\
=&-U_{\DRMST}+O_p\bigl(n^{-1}\bigr).
\label{eq:drmst-expansion}
\end{align*}
 As the remainder is $O_p(n^{-1})$, $Z_{\DRMST}$ and the normalized $\sqrt{n}U_{\DRMST}$ are asymptotically equivalent up to sign.  RAH can be handled the same way, with an extra delta-method step for the ratio $F_k(\tau)/R_k(\tau)$, and
\begin{align*}
\log(\hat{\theta}_{\RAH})=& \int_0^\tau \left[\frac{S_0(\tau)}{F_0(\tau)}
              +\frac{ R_0(\tau)-R_0(t)}{R_0(\tau)}\right]d\hat{\Delta}_\Lambda(t)+O\bigl(\|\hat{\Delta}_\Lambda\|^2\bigr)\\
              =& U_{\RAH}+O_p(n^{-1})
\end{align*}
under local alternative. As a result, $Z_{\RAH}$ and the normalized $\sqrt{n}U_{\RAH}$ are asymptotically equivalent.

\medskip
In summary, under equal allocation and a common censoring assumption, the
resulting weights for $U_{\mathrm Cox}$, $U_{\DRMST}$ and $U_{\RAH}$ are
\begin{align*}
  W_{\mathrm{Cox}}(t)
    &=\;S_0(t)\,S_C(t)/4,\\
  W_{\DRMST}(t)
    &= R_0(\tau)-R_0(t), \mbox{ and}\\
  W_{\RAH}(t)
    &= \left[\frac{S_0(\tau)}{F_0(\tau)}
              +\frac{ R_0(\tau)-R_0(t)}{R_0(\tau)}\right],
\end{align*}
respectively. Since multiplying the weight function by any nonzero constant yields an equivalent test statistic, without loss of generality, we may use following equivalent weight functions
\begin{align*}
  W_{\mathrm{Cox}}(t)=&S_0(t)\,S_C(t)\\
  W_{\DRMST}(t)=&\frac{ R_0(\tau)-R_0(t)}{R_0(\tau)}\\
  \mbox{ and }~
  W_{\RAH}(t)=&S_0(\tau)+\left\{1-S_0(\tau)\right\}W_{\DRMST}(t)
\end{align*}
for three test statistics.
On the current scale, all three weights equal one at $t=0$, their limits at
$t=\tau$ are $S_0(\tau)S_C(\tau)$, $0$, and $S_0(\tau)$. Also, weight for $U_{\RAH}$ is a simple linear transformation of that for $U_{\DRMST}.$ Note that all $U_X$ are oriented so that $U_X<0$ indicates benefit of treatment (hence $U_{\DRMST}=-\hat\theta_{\DRMST}$ to first order), where $X=$ ``Cox", ``D-RMST" and ``RAH".

The Cox HR-based log-rank weight tracks the amount of information in the risk set. It is the largest early in follow-up, when the risk set is large, and decreases over time as deaths and censoring deplete the risk set. Under the PH assumption, this weighting is optimal. The D-RMST weight represents the RMST within the window \([t, \tau]\). It also places greater weight on earlier follow-up, which decreases monotonically to zero at \(\tau\), thereby also down-weighting survival differences in late follow-up. The weight for RAH is likewise monotone decreasing, but remains strictly positive at \(\tau\), and therefore places relatively greater emphasis on late-stage survival differences than the weight function for D-RMST. This common representation allows us to formally examine the test efficiency in the next section.

\subsection{Asymptotic efficiency}
\label{sec:efficiency}

To compare the three tests on a single scale we use the
\emph{Pitman efficacy}, the standard local-power measure for asymptotically
normal tests \citep{Pitman1949}.  Following \citet{Tian2018}, we embed the
two-sample problem in a sequence of contiguous (local) alternatives indexed
by a fixed log-hazard-ratio shape $\alpha(\cdot)$ on $(0,\tau]$,
\begin{equation}\label{eq:local-alt-main}
  H_{An}:\qquad
  \log\!\left\{\frac{\lambda_{1n}(t)}{\lambda_0(t)}\right\}
    =\frac{\alpha(t)}{\sqrt{n}},
\end{equation}
so that the two hazards coincide as $n\to\infty$, while the between-group log-hazard difference is of order $O(n^{-1/2})$.  Under \eqref{eq:local-alt-main}, for
test $X\in\{\mathrm{Cox},\DRMST,\RAH\},$ the test statistic is asymptotically normal, i.e.,
\begin{equation*}
  \sqrt{n}\,U_X\xrightarrow{d}N\!\bigl(\mu_X,\;\sigma_X^2\bigr),
\end{equation*}
where
\begin{equation*}
  \mu_X=\int_0^\tau W_X(t)\,\alpha(t)\,\lambda_0(t)\dd t
  \qquad\mbox{and}\qquad
  \sigma_X^2=\int_0^\tau \frac{W_X(t)^2}{w(t)}\,\lambda_0(t)\dd t .
\end{equation*}
The drift $\mu_X$ is obtained by using the first-order mean increment
$\alpha(t)\lambda_0(t)\dd t$ of
$\sqrt n\,d\hat\Delta_{\Lambda}(t)$ under the local alternative
in \eqref{eq:wlr-generic}. The variance $\sigma_X^2$
is the null variance of
\eqref{eq:wlr-generic}. %
The \emph{asymptotic efficiency} (AE), or Pitman
efficacy, of test $X$ is the squared standardized drift
\begin{equation}\label{eq:efficacy-def}
  \xi_X=\frac{\mu_X^2}{\sigma_X^2},
\end{equation}
which is the signal to noise ratio governing the local power. The derivation of $\mu_X$ for RAH-based test can be found in Appendix \ref{sec:appendix-proofs}. Furthermore, the
explicit expressions of
$\xi_{\mathrm{Cox}},$ $\xi_{\DRMST}$ and $\xi_{\RAH}$ are provided in Appendix \ref{sec:appendix-pitman}.
The \emph{asymptotic relative efficiency} (ARE) of test $X$ versus test $Y$
is the ratio of corresponding AEs:
\begin{equation}\label{eq:are-def}
  \mathrm{ARE}(X,Y)=\frac{\xi_X}{\xi_Y},
\end{equation}
which approximates the inverse ratio of the sample sizes these two tests
require to reach the same local power against \eqref{eq:local-alt-main};
$\mathrm{ARE}(X,Y)>1$ means that test $X$ is more efficient and needs proportionally fewer subjects than $Y$ \citep{Tian2018,Pitman1949}.  Because ARE may be a poor proxy for the sample size ratio in finite-sample setting, we also report the \emph{empirical relative efficiency} (ERE), the
ratio of squared standardized effect sizes estimated by simulation (Section~\ref{sec:sim}),
\begin{equation}\label{eq:ere-def}
  \mathrm{ERE}(X,Y)
    =\frac{E(U_X)^2/\mathrm{Var}(U_X)}
          {E(U_Y)^2/\mathrm{Var}(U_Y)},
\end{equation}
with the expectations and variances approximated by their Monte Carlo averages \citep{Tian2018}.  %

\section{The geometric interpretation of AE}
\label{sec:cosine}
Section~\ref{sec:wlr} reduced all three tests to the weighted logrank test~\eqref{eq:wlr-generic} and showed that their AE depends on the test only through the weight function.  This is the basis
of the classical comparison of \citet{Tian2018}, who contrast the Cox and
D-RMST tests by inspecting their weight functions $W_{\mathrm{Cox}}(\cdot)$ and
$W_{\DRMST}(\cdot)$ directly. In this section, we replace the side-by-side weight comparison with a geometric quantity, a cosine, turning AE into a function of an angle in a function space.

\subsection{Signal direction and test direction}
\label{sec:directions}

Suppose that all candidate statistics are evaluated on the same horizon $\tau$, and consider the function space $L^2([0,\tau], \lambda_0(v)\dd v)$, where the inner product
\begin{equation}\label{eq:inner-product}
  \langle f,g\rangle=\int_0^\tau f(v)\,g(v)\,\lambda_0(v)\dd v,
 \end{equation}
and norm $\|f\|^2=\langle f,f\rangle .$ Next, we
define the \emph{signal direction} for local alternative $\alpha(\cdot)/\sqrt{n}$ and \emph{test direction} for test $X\in\{\mathrm{Cox},\DRMST,\RAH\}$, as
\begin{equation}\label{eq:directions}
  s_\alpha(t)=\alpha(t)\sqrt{w(t)}\qquad \mbox{and} \qquad
  h_X(t)=\frac{W_X(t)}{\sqrt{w(t)}},
\end{equation}
respectively.

The signal direction, $s_\alpha(t),$ describes where the detectable treatment effect is by combining the local log-HR representing the underlying between-group difference at $t$ and $\sqrt{w(t)}$ representing the amount of information available to detect the difference at time $t$.  Thus, $|s_\alpha(t)|$ is large where the true effect is big and sufficient information is available to detect it. It is a ``direction'' since only its shape over $t$, not its overall magnitude, matters. The test direction, $h_X(t)=W_X(t)/\sqrt{w(t)},$ is the analogous
object built from the test's own weight.  $W_X(\cdot)$ records where the test is
prepared to look for the effect. $\sqrt{w(t)}$ in the denominator ensures that $h_X(t)$ represents test weight per unit of information available at time $t$. Next, we have
\[
  \langle s_\alpha,h_X\rangle
    =\int_0^\tau \alpha(t)\sqrt{w(t)}\,\frac{W_X(t)}{\sqrt{w(t)}}\,
      \lambda_0(t)\dd t
    =\int_0^\tau W_X(t)\,\alpha(t)\,\lambda_0(t)\dd t=\mu_X,
\]
\[
\mbox{and}~~~  \|h_X\|^2
    =\int_0^\tau \frac{W_X(v)^2}{w(v)}\,\lambda_0(v)\dd v=\sigma_X^2.
\]
Consequently the AE in \eqref{eq:efficacy-def} can be written as
\begin{equation}\label{eq:eff-inner}
  \xi_X=\frac{\langle s_\alpha,h_X\rangle^2}{\|h_X\|^2}.
\end{equation}

\subsection{The cosine metric}
\label{sec:cosine-metric}

Writing the inner product in \eqref{eq:eff-inner} as

$$
\langle s_\alpha,h_X\rangle=\|s_\alpha\|\,\|h_X\|\cos(\theta_X),
$$
or equivalently,
\begin{equation}\label{eq:eff-cos}
\cos^2(\theta_X)=\frac{\langle s_\alpha,h_X\rangle^2}
{\|s_\alpha\|^2\|h_X\|^2}
=\frac{\xi_X}{\|s_\alpha\|^2},
\end{equation}
where \(\theta_X\) denotes the angle between the signal direction \(s_\alpha\) and the test direction \(h_X\). Intuitively, \(\theta_X\) measures the mismatch between where the test places its weight and where the detectable treatment effect is concentrated: an angle near zero or $\pi$ indicates strong alignment up to sign and high local power. The quantity \(\|s_\alpha\|^2\) is the AE of the oracle test using the optimal weight \(\alpha(t)w(t)\). Therefore, $\cos^2(\theta_X)$ is the ARE of test \(X\) relative to the oracle test, $O$. In general, the ARE of test $X$ vs test $Y$ can be expressed as
\begin{equation}\label{eq:are-cosine}
  \mathrm{ARE}(X, Y)=\frac{\cos^2(\theta_X)}{\cos^2(\theta_Y)}=\frac{\mathrm{ARE}(X, O)}{\mathrm{ARE}(Y, O)}.
\end{equation}

\medskip
As a preview, Figure~\ref{fig:cosine-geometry} provides geometric representations for three data-generating processes considered in \citet{Tian2018}. Specifically, the left column displays \(s_\alpha(t)\), \(h_{\mathrm{Cox}}(t)\), \(h_{\mathrm{DRMST}}(t)\), and \(h_{\mathrm{RAH}}(t)\), whereas the right column depicts the corresponding directions on a compass, with the signal direction fixed at 12 o'clock and each test's \(\cos^2\theta_X\) shown at the tip of the corresponding arrow for test direction. More detailed results are presented in Section~\ref{sec:sim}.

\begin{figure}[ht]
  \centering
  \includegraphics[width=0.95\linewidth]{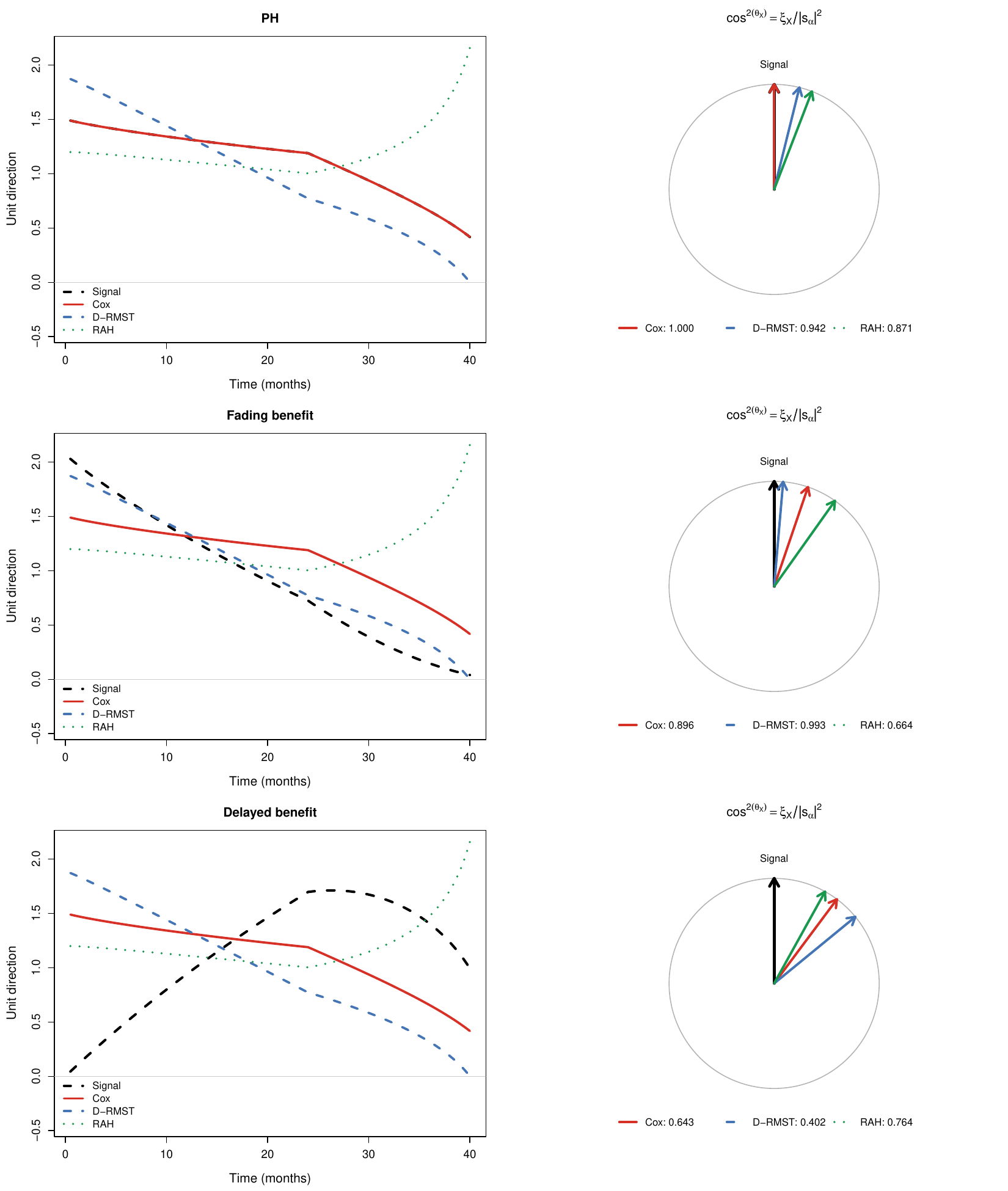}
  \caption{Unit-vector functions (left column) and compass diagrams
  (right column) for three scenarios
  (Weibull baseline: $\kappa=0.826$, $S_0(40)=0.50$, light censoring).
  All three tests are evaluated at
  $\tau=40$ months; %
  Signal direction is placed at 12 o'clock in the
  compass diagrams; each test arrow is annotated with its $\cos^2(\theta_X)$ value.
  Under PH (top), Cox is aligned, while D-RMST and RAH lose some efficiencies.
  Under delayed benefit (bottom), the signal rotates toward the RAH
  direction, raising $\cos^2(\theta_{\RAH})$ above
  $\cos^2(\theta_{\mathrm{Cox}})$.}
  \label{fig:cosine-geometry}
\end{figure}

The cosine metric naturally applies at the trial design stage. Suppose that one has access to:
\begin{itemize}
\item assumed survival function $S_0(t)$ in the control group;
\item assumed treatment-effect $\alpha(t)$ (e.g., PH, delayed
benefit, etc.);
\item planned censoring and accrual distribution (which affects $w(\cdot)$); and
\item a set of candidate test statistics.
\end{itemize}
Given these inputs, one can compute the signal and test directions and $\cos^2(\theta_X)$ for each candidate test and select the test with the highest $\cos^2(\theta_X)$, the one most aligned with the expected signal. Crucially, censoring and accrual are part of trial design and affect the value of $\theta_X.$ By varying $S_C(t)$, one can trace how the cosine metric changes and understand whether the selected test is robust to design assumptions. We demonstrate this in Section~\ref{sec:cm057-analysis} for CheckMate~057, and specifically trace the preferred test across design choices in a sensitivity analysis.

\section{Simulation Study}
\label{sec:sim}

\subsection{Design}
\label{sec:sim-design}

We compare the Cox HR-, D-RMST-, and RAH-based tests via ARE and ERE using the
same data-generating process of \citet{Tian2018}.
The baseline event-time distribution in control group is Weibull($\kappa$, $\eta_0$)
with scale parameter $\eta_0$ chosen so that $S_0(\TauSim)$ equals a specified target survival rate at $\tau=\TauSim$ months and the default shape parameter $\kappa=0.826$ yielding a decreasing hazard
$\lambda_0(t)\propto t^{-0.174}$ in the control group.
Censoring time $C=\min(E,U)$, where $E\sim\mathrm{Exp}(\lambda_c)$ and
$U\sim\mathrm{Uniform}(24,43)$ represent loss of follow-up and random entry, respectively. Two rates $\lambda_c\in\{0.0006,0.025\}$ for $E$
yield light ($S_C(\TauSim)=\ScTauLight$) and heavy ($S_C(\TauSim)=\ScTauHeavy$)
censoring.  Sample size $n_0=n_1=400$ per arm; $B=2000$ replicates are used for each simulation. In the simulations, all three tests use the same analysis horizon $\tau=40$ months; follow-up for the Cox HR-based test is administratively censored at this horizon.  We consider the following alternatives:

\noindent\textbf{PH scenario (Case 1):}
The treatment event-time distribution is Weibull$(\kappa, \eta_1)$ with $\eta_1=\eta_0/0.70^{1/\kappa}$, yielding HR $=0.70$. We adjust $\eta_0$ so that $S_0(40)$ takes values 0.10, 0.30, 0.50, 0.70, and 0.90.\\
\noindent\textbf{Non-PH scenarios:}
$\kappa=0.826$ and
$S_0(40)=0.50$ for event-time from control group. The log-HR $\alpha(t)=\log[\lambda_1(t)/\lambda_0(t)]$
follows two patterns:
\begin{itemize}
  \item \textbf{Case~2 (fading benefit):}
    $\alpha(t)=-c_0\max(\TBFading-t,\,0)^s, s\in \{0.5, 1, 2\}$.
    The treatment benefit is the biggest early and fades to null at month \TBFading.
  \item \textbf{Case~3 (delayed benefit):}
    $\alpha(t)=-c_0t^s,  s\in \{0.5, 1, 2\}$.
    The treatment benefit is absent at $t=0$ and increases over follow-up.
\end{itemize}
In both cases, $c_0>0$ is selected so that the mean log-HR over $[0,\tau]$ is approximately $\log(0.70)$. For Cases 2 and 3,  event-times in treatment group are generated by numerical inversion of the cumulative hazard, i.e., $T_1 \sim \Lambda^{-1}_1\{\mbox{Exp}(1)\},$ where $\Lambda_1(t)=\int_0^t \lambda_0(u)\exp\{\alpha(u)\}du.$

Figures~\ref{fig:dgp-ph}, \ref{fig:dgp-case2}, and~\ref{fig:dgp-case3}
display the survival curves and HRs for three
cases. %
Each figure uses a four-panel layout: the top-left panel shows
$\exp\{\alpha(t)\}$ (the HR), and the
remaining panels show $S_0(t)$ (control, solid) against $S_1(t)$
(treatment, dashed) for different values of $s$.

Three simulation settings are also studied but the results are not reported in the main paper.
\begin{enumerate}
\item the PH setting with a HR=0.82 (Appendix ~\ref{sec:appendix-hr082});
\item the PH setting with alternative baseline survival distributions (Weibull$(\kappa, \eta_0),$ \(\kappa=0.7, 1\) and 1.5), sample size (150 per arm), and trial horizon $(\tau=10)$ month (Appendix~\ref{sec:appendix-tau10});
\item the crossing-hazard setting with $\alpha(t)=\log(0.46)+(t/43)$ (the HR begins at 0.46 and rises to 1.25 at month 43) (Appendix~\ref{sec:appendix-crossing}).
\end{enumerate}

\medskip
\begin{figure}[ht]
  \centering
  \includegraphics[width=0.95\linewidth]{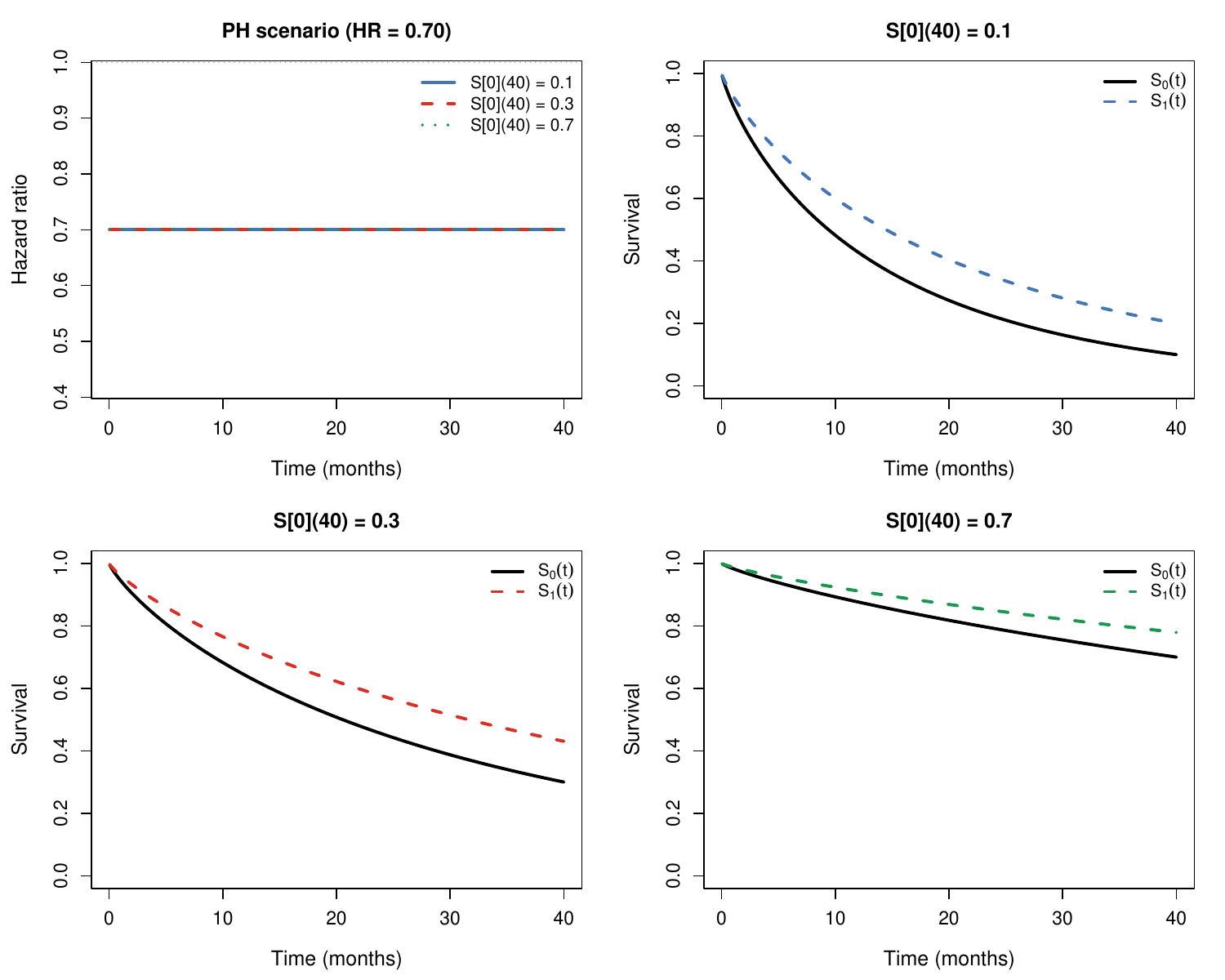}
  \caption{\textbf{PH scenario} (HR $= 0.70$) and
    $S_0(\tau)\in\{0.10, 0.30, 0.70\}$.
    The HR is flat at 0.70; higher event rates produce greater absolute separation between two survival curves.}
  \label{fig:dgp-ph}
\end{figure}

\begin{figure}[ht]
  \centering
  \includegraphics[width=0.95\linewidth]{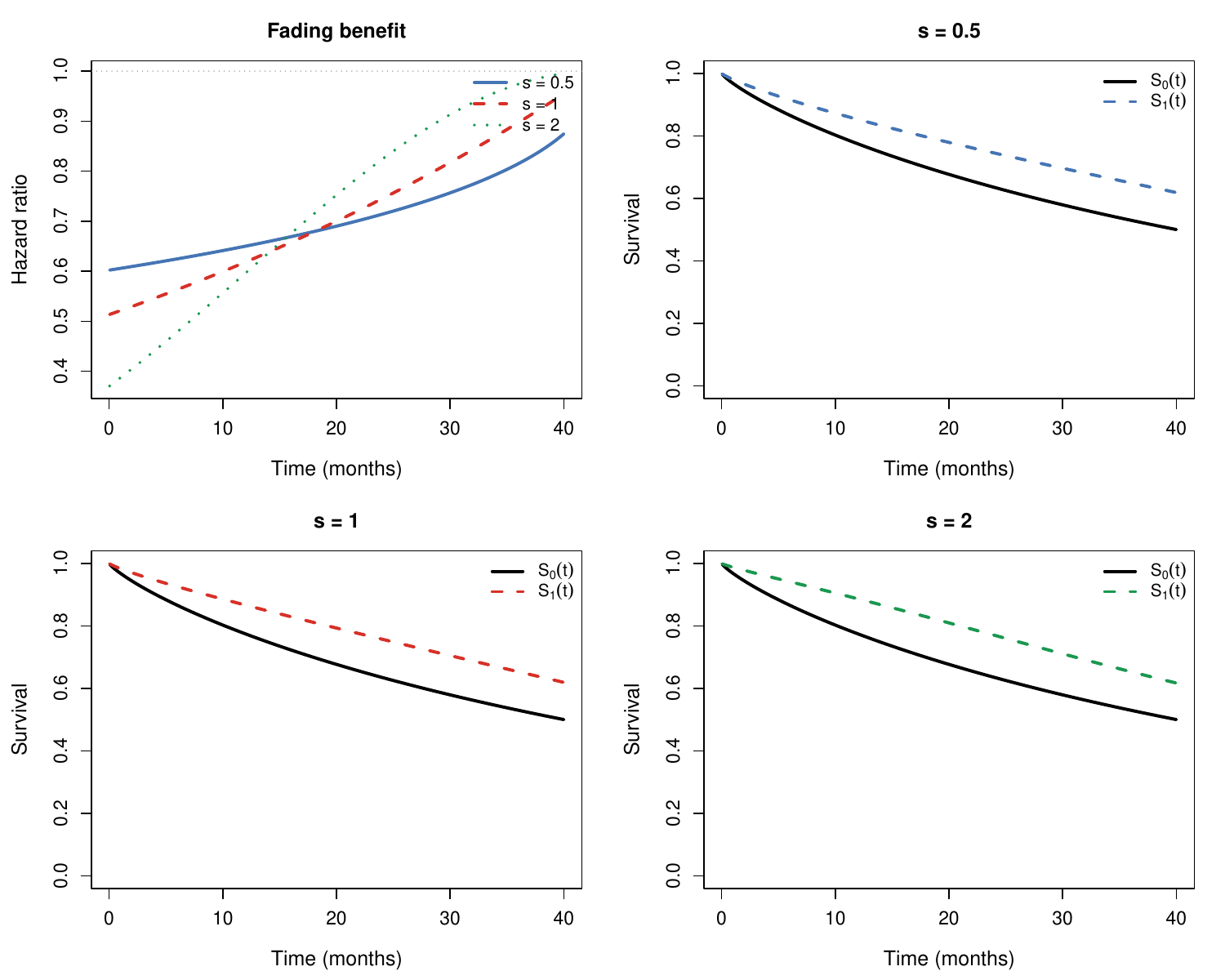}
  \caption{\textbf{Case~2: fading benefit.}
    $\alpha(u)=-c_0\,\max(t_M-u,0)^s$,  mean log-HR $=\log(0.70)$.
    The HR is below 1 early and converges to 1 by $t_M$; larger $s$
    concentrates the benefit closer to $u=0$.}
  \label{fig:dgp-case2}
\end{figure}

\begin{figure}[ht]
  \centering
  \includegraphics[width=0.95\linewidth]{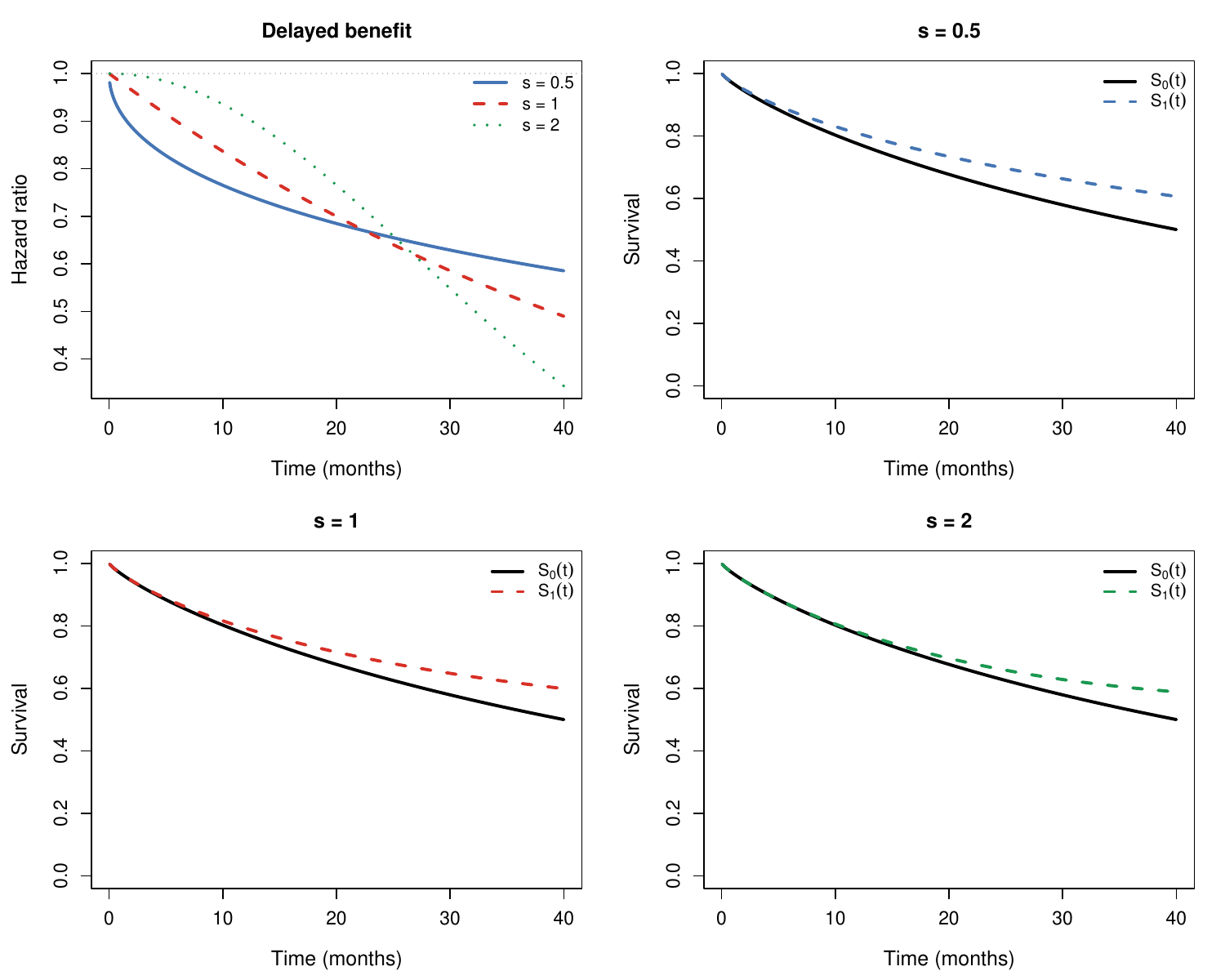}
  \caption{\textbf{Case~3: delayed benefit.}
    $\alpha(u)=-c_0\,u^s$, mean log-HR $=\log(0.70)$.
    The HR starts at 1 and decreases; larger $s$ delays the onset of the
    treatment effect, with the benefit concentrated at later time points for $s=2$.}
  \label{fig:dgp-case3}
\end{figure}

\subsection{PH scenario (Case 1)}
\label{sec:sim-ph}

Table~\ref{tab:main-sim-ph-hr070-s0-010-090-light-heavy} summarizes the comparison for \(S_0(40)\in\{0.10,0.30,0.50,0.70,0.90\}\). Under light censoring, the ARE of the D-RMST-based test relative to Cox HR-based test decreases gradually as \(S_0(40)\) increases. Under heavy censoring, however, this ARE remains close to one across the full range of \(S_0(40)\).

The ARE of the RAH-based test relative to the Cox HR-based test is also below one throughout. This is expected as the Cox HR-based test is optimal under the PH assumption. The efficiency declines as \(S_0(40)\) increases and more strongly under heavy censoring. As \(S_0(\tau)\) increases, the lower bound
$W_{\RAH}(\tau)=S_0(\tau)$ becomes larger, so that $
h_{\RAH}(\tau)=2
\sqrt{S_0(\tau)/S_C(\tau)}$
also increases. This inflates \(\|h_{\RAH}\|^2\) and enlarges RAH's discrepancy from the optimal test direction. At \(S_0(40)=0.10\) under heavy censoring, ARE(RAH/Cox) is \AreRahCoxPHHeavyLowS, whereas ARE(RAH/D-RMST) is \AreRahDrmstPHHeavyLowS. The ordering
$$
\mathrm{Cox}\ge \RAH \approx \DRMST
$$
under PH with light censoring and small \(S_0(\tau)\) is also evident in Figure~\ref{fig:cosine-geometry} (top row): the signal direction is aligned with the Cox HR-based test direction, whereas the RAH and D-RMST test directions are similarly off-axis.

The ARE(RAH/D-RMST) is close to one under light censoring but substantially lower under heavy censoring, indicating a more pronounced performance separation between the two procedures.

Table~\ref{tab:main-sim-ph-hr070-kappa-050-200-light-heavy} reports results for \(S_0(40)=0.50\) and \(\kappa\in\{0.5,0.826,1,1.5,2\}\). Under light censoring, the ARE(RAH/Cox) decreases from \AreRahCoxPHShapeHalf\ when \(\kappa=0.5\), corresponding to a decreasing baseline hazard, to \AreRahCoxPHShapeTwo\ when \(\kappa=2\), corresponding to an increasing baseline hazard. Intuitively, an increasing hazard shifts more events toward later follow-up, where RAH weighting differs more substantially from optimal weighting of Cox HR-based test, thereby reducing its efficiency. Both Tables~\ref{tab:main-sim-ph-hr070-s0-010-090-light-heavy} and~\ref{tab:main-sim-ph-hr070-kappa-050-200-light-heavy} show systematically larger efficiency losses for the RAH-based test in comparison with the optimal Cox HR-based test under heavy censoring than under light censoring. This pattern persists across alternative baseline survival distributions and trial horizons. %

In summary, these results show that the ARE of the competing tests depends materially on the censoring distribution. The censoring and accrual patterns induced by the study design should not be viewed merely as nuisance features; rather, they are important design inputs for choosing appropriate testing procedures.

\begin{table}[ht]
\centering
\caption{ARE and ERE under PH (HR$=0.70$), $\kappa=0.826$, $\tau=40$, ${\tau_{\rm Cox}=40},$ $n=400$ per arm, $B=2000$.}
\label{tab:main-sim-ph-hr070-s0-010-090-light-heavy}
\resizebox{\textwidth}{!}{%
\begin{tabular}{lccccccc}
\toprule
Censoring & $S_0(40)$ & ARE(D/C) & ARE(R/C) & ARE(R/D) & ERE(D/C) & ERE(R/C) & ERE(R/D)
\\
\midrule
Light & 0.10 & 0.994 & 0.941 & 0.947 & 1.006 & 0.946 & 0.940 \\
Light & 0.30 & 0.967 & 0.898 & 0.929 & 0.978 & 0.860 & 0.879 \\
Light & 0.50 & 0.942 & 0.871 & 0.924 & 0.948 & 0.893 & 0.942 \\
Light & 0.70 & 0.920 & 0.851 & 0.924 & 0.912 & 0.866 & 0.950 \\
Light & 0.90 & 0.901 & 0.835 & 0.927 & 0.919 & 0.832 & 0.905 \\
\midrule
Heavy & 0.10 & 0.982 & 0.798 & 0.812 & 0.980 & 0.725 & 0.739 \\
Heavy & 0.30 & 0.997 & 0.720 & 0.722 & 1.019 & 0.691 & 0.678 \\
Heavy & 0.50 & 0.998 & 0.680 & 0.682 & 1.024 & 0.671 & 0.656 \\
Heavy & 0.70 & 0.995 & 0.654 & 0.658 & 1.011 & 0.677 & 0.670 \\
Heavy & 0.90 & 0.990 & 0.635 & 0.642 & 1.019 & 0.658 & 0.646 \\
\bottomrule
\end{tabular}
}
\end{table}

\begin{table}[ht]
\centering
\caption{ARE and ERE under PH (HR$=0.70$), $S_0(40)=0.50$, $\tau=40$, ${\tau_{\rm Cox}=40},$ $n=400$ per arm, $B=2000$.}
\label{tab:main-sim-ph-hr070-kappa-050-200-light-heavy}
\resizebox{\textwidth}{!}{%
\begin{tabular}{lccccccc}
\toprule
Censoring & $\kappa$ & ARE(D/C) & ARE(R/C) & ARE(R/D) & ERE(D/C) & ERE(R/C) & ERE(R/D)
\\
\midrule
Light & 0.500 & 0.962 & 0.922 & 0.959 & 0.970 & 0.914 & 0.942 \\
Light & 0.826 & 0.942 & 0.871 & 0.924 & 0.948 & 0.893 & 0.942 \\
Light & 1.000 & 0.934 & 0.844 & 0.904 & 0.953 & 0.838 & 0.879 \\
Light & 1.500 & 0.919 & 0.781 & 0.850 & 0.936 & 0.779 & 0.832 \\
Light & 2.000 & 0.914 & 0.737 & 0.806 & 0.917 & 0.738 & 0.805 \\
\midrule
Heavy & 0.500 & 0.996 & 0.750 & 0.753 & 1.016 & 0.737 & 0.725 \\
Heavy & 0.826 & 0.998 & 0.680 & 0.682 & 1.024 & 0.671 & 0.656 \\
Heavy & 1.000 & 0.998 & 0.653 & 0.654 & 1.025 & 0.643 & 0.628 \\
Heavy & 1.500 & 0.994 & 0.599 & 0.603 & 1.015 & 0.577 & 0.569 \\
Heavy & 2.000 & 0.991 & 0.570 & 0.576 & 1.029 & 0.536 & 0.521 \\
\bottomrule
\end{tabular}
}
\end{table}

\subsection{Non-PH: fading benefit (Case 2)}
\label{sec:sim-fading}

Table~\ref{tab:main-sim-nonph} presents the two monotone non-PH scenarios with $\kappa=0.826$, $S_0(40)=0.50$, and $s=1$. Under Case~2, the D-RMST-based test outperforms the RAH-based test, with the Cox HR-based test performing intermediately, i.e.,
$$ \DRMST\ge \rm{Cox}\ge \RAH.$$
Because the treatment effect is concentrated in the early follow-up period, the down-weighting of late differences by D-RMST causes little efficiency loss. In contrast, the RAH-based test continues to assign positive weight to late follow-up times, where the treatment effect has largely disappeared. This leads to the expected ordering of three tests. The corresponding geometric interpretation is shown in the middle row of Figure~\ref{fig:cosine-geometry}.

\subsection{Non-PH: delayed benefit (Case 3)}
\label{sec:sim-delayed}

Under Case~3, the ARE of the RAH-based test relative to the Cox HR-based test exceeds 1 under both censoring regimes (Table~\ref{tab:main-sim-nonph}). The treatment effect is concentrated at late time, whereas the weight for D-RMST approaches zero near $\tau$, effectively down-weighting the events in the most informative time window. The Cox model also assigns decreasing weight to later follow-up times, as its weight function $w(t)$ declines, but this down-weighting is less pronounced than that of D-RMST. In contrast, the RAH-based test maintains a bigger positive weight near $\tau$, as $$W_{\RAH}(\tau)=S_0(\tau)>W_{\rm Cox}(\tau)=S_0(\tau)S_C(\tau)>W_{\DRMST}(\tau)=0.$$ Therefore, it preserves a greater alignment with the delayed treatment effect, and yields the expected ordering,
$$\RAH \ge \rm{Cox} \ge  \DRMST.$$
See Figure~\ref{fig:cosine-geometry} (bottom row) for the corresponding geometric interpretation.

Across all reported simulations, ARE and ERE generally favor the same test. Disagreements occur only when both measures are close to one, where Monte Carlo variation and finite-alternative effects can reverse the efficiency ordering. Even in these cases, ARE and ERE remain quantitatively similar.

\begin{table}[ht]
\centering
\caption{ARE and ERE for monotone non-PH patterns; $\tau=40$, {$\tau_{\rm Cox}=40,$} $S_0(40)=0.50$, $n=400$ per arm, $B=2000$.}
\label{tab:main-sim-nonph}
\resizebox{\textwidth}{!}{%
\begin{tabular}{lccccccc}
\toprule
Pattern & Cens. & ARE(D/C) & ARE(R/C) & ARE(R/D) & ERE(D/C) & ERE(R/C) & ERE(R/D)
\\
\midrule
fading & Light & 1.108 & 0.741 & 0.669 & 1.117 & 0.707 & 0.633 \\
fading & Heavy & 1.011 & 0.499 & 0.494 & 1.039 & 0.472 & 0.454 \\
delayed & Light & 0.626 & 1.190 & 1.901 & 0.638 & 1.179 & 1.850 \\
delayed & Heavy & 0.962 & 1.349 & 1.402 & 0.973 & 1.375 & 1.413 \\
\bottomrule
\end{tabular}
}
\end{table}

\section{CheckMate 057: Analysis and Sensitivity}
\label{sec:cm057-analysis}

In this section, we illustrate the practical performance of three two-sample tests using data from the CheckMate 057 trial
\citep{Borghaei2015}. %
To this end, we fit Weibull model to the trial data, compute the
theoretical ARE for each test, and conduct simulations to compute the empirical power and required sample sizes under the fitted parametric model.

\subsection{Trial background and survival model}
\label{sec:cm057-background}

CheckMate 057 was a randomized, open-label, Phase III trial comparing
nivolumab with docetaxel in 582 patients with previously treated
non-squamous non-small-cell lung cancer; the primary endpoint was
overall survival \citep{Borghaei2015}.  Because individual-level data
are not publicly available, we reconstructed them from the published
KM curves and the reported risk-set information using the
algorithm of \citet{Guyot2012}, yielding $n_0=290$ docetaxel and
$n_1=292$ nivolumab patients with 419 observed deaths (226 and 193,
respectively).

Throughout this section, all three tests use the same analysis horizon
$\tau=24$ months, corresponding to a landmark time of 2 years, which is within the observed follow-up of both groups.  The D-RMST and RAH test statistics are defined on $(0,\tau]$ by
construction; for the Cox HR-based test, follow-up is administratively
censored at $\tau$, so that all three tests examine the treatment effect on the same time window.

The Schoenfeld global test yields $p = \SchoenfeldPCM$, providing strong evidence against the PH assumption. We therefore fit separate Weibull models to each arm and obtain:
\begin{align*}
  \text{Docetaxel (arm 0):} &\quad \kappa_0=\WeibullShapeControl,\; \eta_0=\WeibullScaleControl\text{ months}, \\
  \text{Nivolumab (arm 1):} &\quad \kappa_1=\WeibullShapeTreatment,\; \eta_1=\WeibullScaleTreatment\text{ months}.
\end{align*}
Figure~\ref{fig:weibull} overlays the nonparametric KM curves with the fitted Weibull curves, demonstrating their close agreement.

\begin{figure}[ht]
  \centering
  \includegraphics[width=0.95\textwidth]{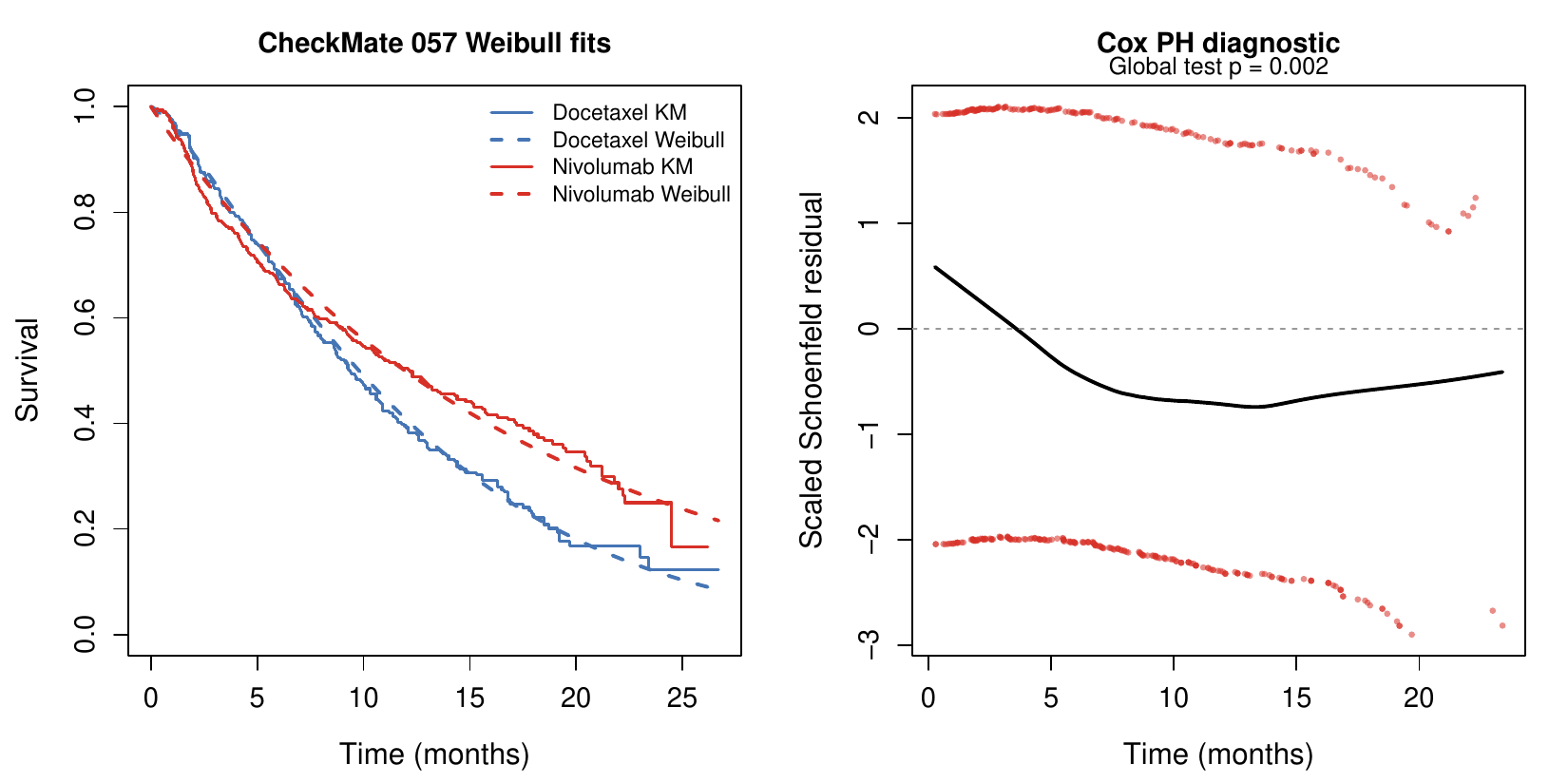}
  \caption{Survival model fit (left) and proportional-hazards diagnostic (right)
    for the CheckMate 057 trial.}
  \label{fig:weibull}
  \label{fig:ph-diag}
  \label{fig:cm057-combined}
\end{figure}

In the docetaxel arm, $\kappa_0>1$ induces a monotone increasing hazard. In the nivolumab arm, $\kappa_1\approx1$ indicates an approximately constant hazard. The fitted HR is time varying: $\exp\{\alpha(1)\}=1.206$, it crosses one near month two, and it reaches 0.500 at month 24. This pattern is characteristic of immunotherapy and is consistent with the Schoenfeld diagnosis (Figure~\ref{fig:ph-diag}), whose negative residual slope indicates a decreasing log HR over time.

For the design-stage evaluations in terms of ARE, ERE, power, and required sample size, we pair the fitted Weibull event-time model with a working model for censoring informed by the trial design. Patients were enrolled between November 2012 and December
2013, and overall survival was assessed at a database lock in March 2015,
with a reported minimum follow-up of 13.2 months \citep{Borghaei2015}.
Assuming approximately uniform accrual, we use the rounded administrative
follow-up window $U\sim\mathrm{Uniform}(\CMAdminLower,\CMAdminUpper)$ months.
We add an independent exponential dropout time
$E\sim\mathrm{Exp}(\lambda_c)$ with $\lambda_c=\CMDropoutRate$ per month
(about $\CMDropoutAnnualPct\%$ per year), and set $C=\min(U,E)$.
The administrative window is consistent with the predominantly late
censoring in the reconstructed data; the dropout rate is treated as a design assumption and examined in the sensitivity analysis below.

Table~\ref{tab:cm057-are-ere} presents the ARE and ERE results. Both measures yield the same ordering of the three tests:
\[
\RAH \ge \mathrm{Cox} \ge \DRMST.
\]
Figure~\ref{fig:cosine-cm057} presents the compass diagram for the CheckMate 057 data. The corresponding squared cosine similarities are
\[
\cos^2(\theta_{\mathrm{Cox}})=\CosSqCoxCM,\qquad
\cos^2(\theta_{\DRMST})=\CosSqDrmstCM,\qquad
\cos^2(\theta_{\RAH})=\CosSqRahCM.
\]
Because the ARE between two tests is given by the ratio of their squared cosine similarities, these values imply
\[
\mathrm{ARE}(\RAH/\mathrm{Cox})=\AreRahCoxCM
\quad\text{and}\quad
\mathrm{ARE}(\RAH/\DRMST)=\AreRahDrmstCM.
\]

Figure~\ref{fig:cosine-cm057} also provides a geometric explanation for why this alternative favors the RAH-based test. Although the fitted HR crosses one, it does so within the first two months, so the treatment benefit occupies essentially the entire analysis window and is strongest late in follow-up. The signal direction places substantial weight on the late follow-up period, where the treatment benefit of nivolumab is greatest. In contrast, the Cox HR-based test direction places more weight on earlier times, which gradually decays over follow-up. The D-RMST test direction tapers more rapidly to zero as time approaches $\tau$, thereby placing even less weight on the late follow-up period. The RAH test weight, by comparison, is bounded below by $S_0(\tau)\approx\SzeroCMTau$, allowing it to retain meaningful weight at later times. Consequently, the RAH test direction is more closely aligned with the signal direction, resulting in greater efficiency under this late-emerging benefit.

\begin{table}[ht]
\centering
\caption{ARE and ERE under parametric model fitted to the CheckMate 057 trial; all methods use $\tau=24$ months, $n=300$ per arm, $B=500$.}
\label{tab:cm057-are-ere}
\begin{tabular}{lcc}
\toprule
Comparison & ARE & ERE
\\
\midrule
D-RMST/Cox & 0.778 & 0.776 \\
RAH/Cox & 1.067 & 1.167 \\
RAH/D-RMST & 1.372 & 1.504 \\
\bottomrule
\end{tabular}
\end{table}

\begin{figure}[ht]
  \centering
  \includegraphics[width=0.90\linewidth]{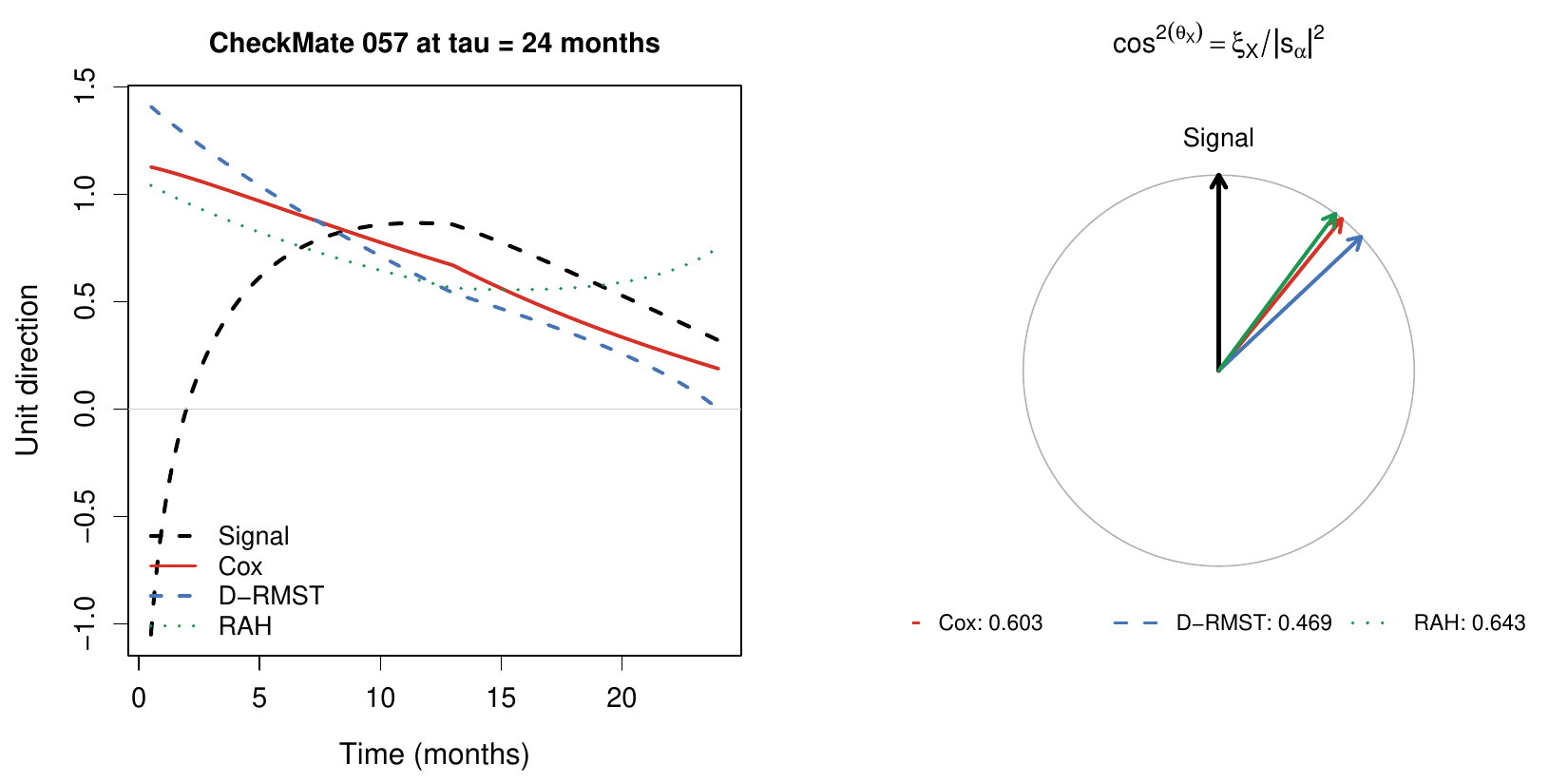}
  \caption{Unit-vector functions (left) and compass diagram (right)
  for Weibull distributions fitted to CheckMate 057 trial.
  All three tests use $\tau=\TauCM$ months; the function panel displays them on $[0,\tau]$.
  The signal direction is placed at 12 o'clock and reflects the
  fitted non-PH pattern, in which the nivolumab benefit emerges late in follow-up.
  Each test direction is annotated with
  $\cos^2(\theta_X)$.
  RAH is closest to the signal direction, followed by Cox and D-RMST.}
  \label{fig:cosine-cm057}
\end{figure}

Beyond the efficiency metrics themselves, a key practical question is how differences in ARE translate into differences in power and required sample size. Figure~\ref{fig:power} shows the empirical power as a function of sample size per arm, with the RAH-based test consistently achieving the highest power, followed by the Cox HR-based test and then the D-RMST-based test. Table~\ref{tab:cm057-power-n80} summarizes power at $n_0=n_1=300$ per arm and the smallest sample size on the prespecified 50-patient incremental grid that reaches 80\% power based on simulation. RAH reaches the target power at \NPowerRahCM\ patients per arm, compared with \NPowerCoxCM\ for Cox and \NPowerDrmstCM\ for D-RMST. Thus, Cox and D-RMST require \NPowerCoxMorePct\% and \NPowerDrmstMorePct\% more patients than RAH, respectively. At 300 patients per arm, the empirical power is \PowerRahCMPct\% for RAH-based test, \PowerCoxCMPct\% for Cox HR-based test, and \PowerDrmstCMPct\% for D-RMST-based test.

\begin{table}[ht]
\centering
\caption{Empirical power at $n=300$ per arm and the smallest sample size on  50-patient incremental grid that reaches 80\% power based on simulation ( $B=2000$). The ratio of required sample size for 80\% power based on simulation, ARE and ERE.}
\label{tab:cm057-power-n80}
\begin{tabular}{lccccc}
\toprule
Method & Power & $n_{80}$ & Simulation & ARE & ERE
\\
\midrule
Cox & 0.849 & 300 & 1.20 & 1.07 & 1.17 \\
D-RMST & 0.725 & 400 & 1.60 & 1.37 & 1.50 \\
RAH & 0.889 & 250 & 1.00 & 1.00 & 1.00 \\
\bottomrule
\end{tabular}
\end{table}

\begin{figure}[ht]
  \centering
  \includegraphics[width=0.65\textwidth]{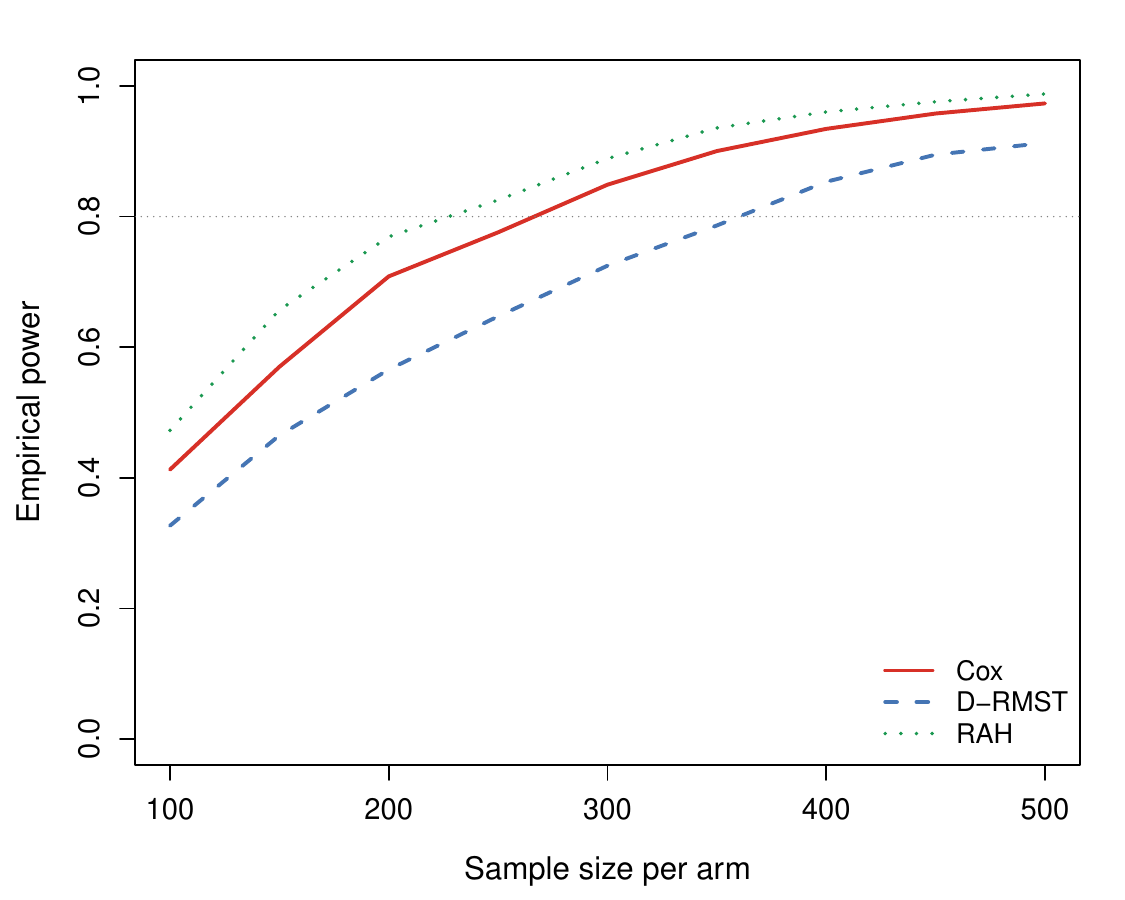}
  \caption{Empirical power for Cox (red, solid), D-RMST (blue, dashed), and RAH (green, dotted) under the fitted CheckMate 057 non-PH alternative. Horizontal line: 80\% power.}
  \label{fig:power}
\end{figure}

We now assess the robustness of the ordering $\RAH \ge \rm{Cox} \ge \DRMST$
to the analysis horizon~$\tau$, dropout rate~$\lambda_c$, and assumed alternative. The results are reported in  Table~\ref{tab:sens-are}.
\begin{enumerate}
\item {\bf Horizon $\tau$:}
ARE(D-RMST/Cox) increases substantially with $\tau$, from \AreDrmstCMTauEighteen\ at $\tau=18$ to \AreDrmstCMTauTwentyfour\ at $\tau=24$. The relatively low efficiency of D-RMST at shorter horizons reflects truncation of the late treatment benefit. In contrast, ARE(RAH/Cox) is comparatively stable (\AreRahCMTauEighteen, \AreRahCMTauTwentyone and \AreRahCMTauTwentyfour\ at 18, 21 and 24 months, respectively), while ARE(RAH/D-RMST) decreases from \AreRahDrmstCMTauEighteen\ to \AreRahDrmstCMTauTwentyfour over the same time window.

\item {\bf Dropout rate:}
Increasing $\lambda_c$ from 0.001 to 0.020 shifts ARE(D-RMST/Cox) from \AreDrmstCMDropoutLow\ to \AreDrmstCMDropoutHigh, indicating that heavier dropout reduces the relative efficiency of Cox HR-based test. ARE(RAH/Cox) increases modestly from \AreRahCMDropoutLow\ to \AreRahCMDropoutHigh, whereas ARE(RAH/D-RMST) decreases from \AreRahDrmstCMDropoutLow\ to \AreRahDrmstCMDropoutHigh. This pattern reflects the growing attenuation of the weight function for Cox HR-based test at later times as $S_C(t)$ decreases.

\item {\bf PH vs.\ non-PH:}
Under a PH alternative (HR$=0.75$, with the same Weibull shape as the event-time distribution for docetaxel arm), Cox HR-based test is the most efficient, with ARE(D-RMST/Cox)$=\AreDrmstCMPH$ and ARE(RAH/Cox)$=\AreRahCMPH$. This highlights that RAH's advantage in the assumed non-PH setting arises from its greater alignment with the late treatment benefit; the positive weight floor $W_{\RAH}(\tau)=S_0(\tau)>0$ provides no such advantage under PH.
\end{enumerate}

In summary, the ERE and ARE results are generally consistent. Under the assumed non-PH alternative, the RAH-based test is more efficient than both Cox HR- and D-RMST-based tests. Under the PH alternative, however, the RAH-based test is slightly less efficient than the Cox HR-based test.

\begin{table}[ht]
\centering
\caption{Sensitivity to analysis horizon, dropout rate, and hazard structure; the common baseline is $\tau=24$, $\lambda_c=0.005$, $n=300$ per arm and $B=500$.}
\label{tab:sens-are}
\resizebox{\textwidth}{!}{%
\begin{tabular}{lccccccc}
\toprule
Panel & Setting & ARE(D/C) & ARE(R/C) & ARE(R/D) & ERE(D/C) & ERE(R/C) & ERE(R/D)
\\
\midrule
Horizon & $\tau=18$ & 0.534 & 1.063 & 1.990 & 0.458 & 1.102 & 2.407 \\
Horizon & $\tau=21$ & 0.657 & 1.080 & 1.644 & 0.613 & 1.126 & 1.838 \\
Horizon & $\tau=24$ & 0.778 & 1.067 & 1.372 & 0.776 & 1.167 & 1.504 \\
Dropout & $\lambda_c=0.001$ & 0.749 & 1.047 & 1.398 & 0.723 & 1.100 & 1.521 \\
Dropout & $\lambda_c=0.005$ & 0.778 & 1.067 & 1.372 & 0.776 & 1.167 & 1.504 \\
Dropout & $\lambda_c=0.020$ & 0.891 & 1.136 & 1.274 & 0.876 & 1.206 & 1.377 \\
Hazard & Non-PH (fitted) & 0.778 & 1.067 & 1.372 & 0.776 & 1.167 & 1.504 \\
Hazard & PH (HR$=0.75$) & 0.988 & 0.909 & 0.919 & 0.938 & 0.926 & 0.988 \\
\bottomrule
\end{tabular}
}
\end{table}

\section{Discussion}
\label{sec:discussion}

The results provide a unified view of the relative efficiency of the Cox-, D-RMST-, and RAH-based tests. A central finding is that no single test is uniformly most efficient; rather, relative efficiency depends jointly on the pattern of treatment effect and the trial design.
The Cox HR-based test is most efficient under the PH assumption. Under non-PH alternatives, the D-RMST-based test can be more efficient when treatment effects are stronger earlier in follow-up, whereas the RAH-based test can be more efficient when treatment benefits emerge later. These patterns, however, are not determined by the treatment-effect pattern alone. Censoring and accrual affect the information available over follow-up and can alter the relative test efficiency, and in some settings reverse, the efficiency ordering among the tests.

The proposed cosine metric provides a geometric way to understand these differences by quantifying the alignment between the treatment-effect signal and the direction targeted by each test.  At the design stage, the metric can be used to evaluate candidate tests across a range of plausible assumptions about the treatment-effect pattern and other design features. Such scenario-based comparisons can inform the prespecified choice of the primary test and assess how robust that choice is to uncertainty in the design assumptions.

Several extensions warrant future investigation. The cosine metric could potentially be used not only to compare candidate tests but also as an objective function for optimizing design features such as accrual and follow-up duration.  Another extension is to incorporate the RAH weight $W_{\RAH}(\cdot),$ which is not a member of the Fleming--Harrington weight family \citep{FlemingHarrington1991}, into a Max-combo test \citep{LinEtal2020} together with $W_{\mathrm{Cox}}(\cdot)$ and $W_{\DRMST}(\cdot).$ Such a procedure could provide sensitivity to a broader range of treatment-effect patterns, although its operating characteristics would require further study.

\appendix

\section{Asymptotic Distribution of the RAH Statistic}
\label{sec:appendix-proofs}

\subsection{Local-alternative drift}
Throughout we work under \emph{Pitman contiguous alternatives}:
\begin{equation}\label{eq:local-alt}
  H_{An}:\quad
  \log\frac{\lambda_{1n}(t)}{\lambda_0(t)}=\frac{\alpha(t)}{\sqrt{n}},
\end{equation}
where $\alpha(\cdot)$ is a fixed function on $(0,\tau]$.
Under~\eqref{eq:local-alt}, to first order in $n^{-1/2}$:
$\sqrt{n}[S_{1n}(t)-S_0(t)]\to-S_0(t)\int_0^t\lambda_0(u)\alpha(u)\dd u$.

\begin{proposition}\label{prop:mean-rah}
Under $H_{An}$, $\sqrt{n}\log(\hat\theta_{\RAH})$ is asymptotically normal
with mean
\begin{equation}\label{eq:mean-RAH-compact}
  \mu_{\RAH}
    =\int_0^\tau
       \!\left[\frac{1}{F_0(\tau)}-\frac{R_0(t)}{R_0(\tau)}\right]
       \alpha(t)\lambda_0(t)\dd t.
\end{equation}
\end{proposition}

\begin{proof}
We decompose $\sqrt{n}\log(\hat\theta_{\RAH})=U_n+V_n$, where
\begin{align*}
U_n&=\sqrt{n}[\log\hat F_1(\tau)-\log\hat F_0(\tau)],\\
V_n&=-\sqrt{n}[\log\hat R_1(\tau)-\log\hat R_0(\tau)].
\end{align*}
Next, we have
\begin{align*}
  U_n &= F_0(\tau)^{-1}\sqrt{n}[\hat F_1-\hat F_0](\tau)+o_p(1)\\
      &=F_0(\tau)^{-1}\sqrt{n}[F_1-F_0](\tau)\\
      &\quad-F_0(\tau)^{-1}\sqrt{n}[(\hat S_1- S_1)-(\hat S_0-S_0)](\tau)+o_p(1)\\
      &=\int_0^\tau\frac{S_0(\tau)}{F_0(\tau)}\alpha(t)\lambda_0(t)\dd t+\eta_{1n}+o_p(1),
\end{align*}
and
\begin{align*}
V_n &= -R_0(\tau)^{-1}\int_0^\tau\sqrt{n}[\hat S_1(t)-\hat S_0(t)]\dd t +o_p(1)\\
    &= -R_0(\tau)^{-1}\int_0^\tau\sqrt{n}[S_1(t)-S_0(t)]\dd t\\
    &\quad-R_0(\tau)^{-1}\int_0^\tau\sqrt{n}\left[(\hat{S}_1-S_1)-(\hat{S}_0-S_0)\right](t)\dd t+o_p(1)\\
    &=R_0(\tau)^{-1}\int_0^\tau S_0(t)A(t)\dd t+\eta_{2n}+o_p(1),
\end{align*}
where $A(t)=\int_0^t\lambda_0(u)\alpha(u)\dd u$ and $(\eta_{1n}, \eta_{2n})'$ converging in distribution to a mean zero bivariate Gaussian distribution as $n\to \infty.$  Furthermore,
\begin{align*}
R_0(\tau)^{-1}\int_0^\tau S_0(t) A(t)\dd t
&=A(\tau)-\int_0^\tau\lambda_0(t)\alpha(t)\frac{R_0(t)}{R_0(\tau)}\dd t\\
&=\int_0^\tau\left[1-\frac{R_0(t)}{R_0(\tau)}\right]\alpha(t)\lambda_0(t)\dd t.
\end{align*}
Substituting and simplifying gives~\eqref{eq:mean-RAH-compact}.
\end{proof}

\subsection{Pitman efficacy of RAH-based test}
\label{sec:appendix-pitman}
The Pitman efficacy formula follows from $\xi_{\RAH}=\mu_{\RAH}^2/\sigma_{\RAH}^2$:
\begin{equation}\label{eq:eff-RAH}
  \xi_{\RAH}
    =\frac{\displaystyle
      \left[\int_0^\tau\!\left\{\dfrac{1}{F_0(\tau)}-\dfrac{R_0(t)}{R_0(\tau)}\right\}
            \alpha(t)\lambda_0(t)\dd t\right]^{\!2}}{
      \displaystyle
      \int_0^\tau\!\left\{\dfrac{1}{F_0(\tau)}-\dfrac{R_0(t)}{R_0(\tau)}\right\}^{\!2}
      \dfrac{\lambda_0(t)}{w(t)}\dd t}.
\end{equation}
The efficacies of Cox and D-RMST follow from \citet{Tian2018}:
\begin{equation}\label{eq:eff-Cox}
  \xi_{\mathrm{Cox}}
    =\frac{\displaystyle
       \left[\int_0^{\tau} w(t)\,\alpha(t)\lambda_0(t)\dd t\right]^2}{
      \displaystyle
       \int_0^{\tau} w(t)\,\lambda_0(t)\dd t},
\end{equation}
\begin{equation}\label{eq:eff-DRMST}
  \xi_{\DRMST}
    =\frac{\displaystyle
       \left[\int_0^{\tau}\!\left\{R_0(\tau)-R_0(t)\right\}
             \alpha(t)\lambda_0(t)\dd t\right]^2}{
      \displaystyle
       \int_0^{\tau}\!\left\{R_0(\tau)-R_0(t)\right\}^{\!2}
       \frac{\lambda_0(t)}{w(t)}\dd t}.
\end{equation}

\section{Derivation of the Weight Function Representations}
\label{sec:appendix-weights}

We prove that each of the three statistics of
Section~\ref{sec:three-stats} is asymptotically equivalent to a version of
weighted logrank statistic~\eqref{eq:wlr-generic},
$$\int_0^\tau W(t)\,d[\hat\Lambda_1(t)-\hat\Lambda_0(t)].$$
All derivations use: (i) the KM linear representation
$$\sqrt{n}(\hat S_k-S_k)=-S_k\sqrt{n}(\hat\Lambda_k-\Lambda_k)+o_p(1);$$
(ii)  the functional delta method; and
(iii) the local-alternative approximation
$$S_{1n}(t)=S_0(t)-\frac{S_0(t)A(t)}{\sqrt{n}}+o(n^{-1/2}).$$
For simplicity, we also assume equal allocation and common censoring $S_C(\cdot)$.

\begin{itemize}
\item \textbf{Cox HR-based statistic.}
  Let
\(\ell_n(\theta)\) be the log partial likelihood over \([0, \tau]\),
\(U_n(\theta)=\partial\ell_n(\theta)/\partial\theta\) its score, and
\(\mathcal I_n(\theta)=-\partial^2\ell_n(\theta)/\partial\theta^2\) the
observed information.  The score at \(\theta=0\) is exactly the log-rank
statistic, \(n^{-1}U_n(0)=U_{\mathrm{Cox}}\) in ~\eqref{eq:cox-stat}, and
\[
  n^{-1}\mathcal I_n(0)
    \;\longrightarrow\;
  \int_0^{\tau}w(t)\lambda_0(t)dt=\sigma_\theta^{-2}
\]
in probability, uniformly in a neighbourhood of \(\theta=0\).  Since
\(\hat\theta=O_p(n^{-1/2})\) under the local
alternatives~\eqref{eq:local-alt} and \(U_n(\hat\theta)=0\), a one-term
Taylor expansion of \(U_n(\theta)\) about \(\theta=0\) gives
\[
  0=n^{-1}U_n(\hat\theta)
   =U_{\mathrm{Cox}}-n^{-1}\mathcal I_n(0)\,\hat\theta+o_p(n^{-1/2})
  \Rightarrow
  \sqrt n\,\hat\theta=\sigma_\theta^{2}\sqrt n\,U_{\mathrm{Cox}}+o_p(1).
\]
Hence \(\sqrt n\,\hat\theta=\sigma_\theta^2\sqrt n\,U_{\mathrm{Cox}}+o_p(1)\), and \(Z_{\mathrm{Cox}}\) and the normalized log-rank statistic are
asymptotically equivalent.  The logrank statistic~\eqref{eq:cox-stat} is
already of the weighted-log-rank form, with random weight
\[
 \hat{w}_{\rm{Cox}}(t)=\frac{Y_1(t)Y_0(t)}{Y_1(t)+Y_0(t)},
\]
where $Y_k(t)$ is the size of the risk set of group $k$ at $t.$  By law of large numbers, $n^{-1}\hat{w}_{\rm{Cox}}(t)$ converges to \(w(t)\), which is $S_0(t)S_C(t)/4$ under equal
allocation and common censoring. Therefore, we can set
\[
  W_{\mathrm{Cox}}(t)=4w(t)=\;S_0(t)S_C(t),
\]
in (\ref{eq:wlr-generic}).

\item \textbf{D-RMST.}
By functional delta-method, we have
\begin{align*}
  \sqrt{n}\hat\theta_{\DRMST}&=-\sqrt{n}\int_0^\tau S_0(t)\,\hat{\Delta}_\Lambda(t)\dd t+O\bigl(\sqrt{n}\|\hat{\Delta}_\Lambda\|^2\bigr)\\
  &=-\sqrt{n}\int_0^\tau S_0(t)\int_0^t d\hat{\Delta}_\Lambda(u)\dd t+O_p\bigl(n^{-1/2}\bigr)\\
  &=-\sqrt{n}\int_0^\tau\left\{\int_u^\tau S_0(t)\dd t\right\}d\hat{\Delta}_\Lambda(u)+O_p\bigl(n^{-1/2}\bigr)\\
  &=-\sqrt{n}\int_0^\tau\{R_0(\tau)-R_0(u)\}\,d\hat{\Delta}_\Lambda(u)+o_p(1).
\end{align*}
Therefore, we may set the weight as
\[
  W_{\DRMST}(t)=\frac{R_0(\tau)-R_0(t)}{R_0(\tau)},
\]
as $Z_{\DRMST}$ is invariant to multiplying any nonzero constant to the weight function.

\item \textbf{RAH.}
Using the same argument together with the delta method for $\log F_k(\tau)$ and $\log R_k(\tau)$,
\[
\log\hat\theta_{\RAH}=\int_0^\tau\left[\frac{S_0(\tau)}{F_0(\tau)}+\frac{R_0(\tau)-R_0(t)}{R_0(\tau)}\right]d\hat\Delta_\Lambda(t)+O_p(n^{-1}).
\]
Multiplying by $F_0(\tau)>0$ gives the equivalent weight
\[
  W_{\RAH}(t)
    =S_0(\tau)+F_0(\tau)\,\frac{R_0(\tau)-R_0(t)}{R_0(\tau)}
    =S_0(\tau)+\bigl[1-S_0(\tau)\bigr]\, W_{\DRMST}(t).
\]
\end{itemize}

\section{Simulation Study for PH Setting at HR = 0.82}
\label{sec:appendix-hr082}

We repeat the simulation for PH scenarios with HR=0.82. The results are reported in Tables~\ref{tab:sim3a-hr082} and~\ref{tab:sim3c-hr082}.
Both ARE and ERE are similar to those for HR$=0.70$:

\begin{table}[ht]
\centering
\caption{ARE and ERE under PH (HR$=0.82$), $\kappa=0.826$, $\tau=40$, $n_1=n_0=400$ per arm and $B=2000$.}
\label{tab:sim3a-hr082}
\resizebox{\textwidth}{!}{%
\begin{tabular}{lccccccc}
\toprule
Censoring & $S_0(40)$ & ARE(D/C) & ARE(R/C) & ARE(R/D) & ERE(D/C) & ERE(R/C) & ERE(R/D)
\\
\midrule
Light & 0.10 & 0.994 & 0.941 & 0.947 & 1.001 & 0.938 & 0.937 \\
Light & 0.20 & 0.980 & 0.917 & 0.935 & 0.987 & 0.898 & 0.910 \\
Light & 0.30 & 0.967 & 0.898 & 0.929 & 0.968 & 0.911 & 0.942 \\
Light & 0.40 & 0.954 & 0.883 & 0.926 & 0.957 & 0.872 & 0.911 \\
Light & 0.50 & 0.942 & 0.871 & 0.924 & 0.929 & 0.885 & 0.953 \\
Light & 0.60 & 0.931 & 0.860 & 0.924 & 0.958 & 0.852 & 0.889 \\
Light & 0.70 & 0.920 & 0.851 & 0.924 & 0.930 & 0.839 & 0.903 \\
Light & 0.80 & 0.911 & 0.842 & 0.925 & 0.952 & 0.831 & 0.872 \\
Light & 0.90 & 0.901 & 0.835 & 0.927 & 0.913 & 0.834 & 0.914 \\
\midrule
Heavy & 0.10 & 0.982 & 0.798 & 0.812 & 0.990 & 0.731 & 0.738 \\
Heavy & 0.20 & 0.993 & 0.751 & 0.756 & 1.004 & 0.725 & 0.722 \\
Heavy & 0.30 & 0.997 & 0.720 & 0.722 & 1.004 & 0.704 & 0.701 \\
Heavy & 0.40 & 0.998 & 0.698 & 0.699 & 1.004 & 0.689 & 0.686 \\
Heavy & 0.50 & 0.998 & 0.680 & 0.682 & 1.007 & 0.672 & 0.667 \\
Heavy & 0.60 & 0.997 & 0.666 & 0.668 & 1.001 & 0.676 & 0.675 \\
Heavy & 0.70 & 0.995 & 0.654 & 0.658 & 0.996 & 0.685 & 0.688 \\
Heavy & 0.80 & 0.992 & 0.644 & 0.649 & 0.993 & 0.661 & 0.665 \\
Heavy & 0.90 & 0.990 & 0.635 & 0.642 & 1.054 & 0.698 & 0.662 \\
\bottomrule
\end{tabular}
}
\end{table}

\begin{table}[ht]
\centering
\caption{ARE, ERE and simulated power under PH (HR$=0.82$), $S_0(40)=0.50$, $\tau=40$, $n_1=n_0=400$ per arm and $B=2000$.}
\label{tab:sim3c-hr082}
\resizebox{\textwidth}{!}{%
\begin{tabular}{lcccccccccc}
\toprule
Censoring & $\kappa$ & ARE(D/C) & ARE(R/C) & ARE(R/D) & Power C & Power D & Power R & ERE(D/C) & ERE(R/C) & ERE(R/D)
\\
\midrule
Light & 0.500 & 0.962 & 0.922 & 0.959 & 0.459 & 0.452 & 0.433 & 0.980 & 0.918 & 0.937 \\
Light & 0.826 & 0.942 & 0.871 & 0.924 & 0.438 & 0.418 & 0.395 & 0.929 & 0.885 & 0.953 \\
Light & 1.000 & 0.934 & 0.844 & 0.904 & 0.450 & 0.406 & 0.393 & 0.941 & 0.838 & 0.891 \\
Light & 1.500 & 0.919 & 0.781 & 0.850 & 0.416 & 0.382 & 0.341 & 0.931 & 0.802 & 0.862 \\
Light & 2.000 & 0.914 & 0.737 & 0.806 & 0.385 & 0.349 & 0.307 & 0.906 & 0.765 & 0.844 \\
\midrule
Heavy & 0.500 & 0.996 & 0.750 & 0.753 & 0.377 & 0.379 & 0.307 & 1.003 & 0.744 & 0.741 \\
Heavy & 0.826 & 0.998 & 0.680 & 0.682 & 0.338 & 0.337 & 0.260 & 1.007 & 0.672 & 0.667 \\
Heavy & 1.000 & 0.998 & 0.653 & 0.654 & 0.299 & 0.301 & 0.237 & 0.996 & 0.632 & 0.635 \\
Heavy & 1.500 & 0.994 & 0.599 & 0.603 & 0.280 & 0.284 & 0.224 & 0.991 & 0.632 & 0.638 \\
Heavy & 2.000 & 0.991 & 0.570 & 0.576 & 0.251 & 0.253 & 0.183 & 0.995 & 0.584 & 0.587 \\
\bottomrule
\end{tabular}
}
\end{table}

\section{PH Simulations with Alternative Designs}
\label{sec:appendix-tau10}

We repeat the simulation study for the PH setting under different design:
$\tau=10$, $n_0=n_1=150$, and three baseline distributions (
an exponential with $S_0(10)\approx0.37$;
a Weibull baseline with a
decreasing hazard. ($\kappa=0.7$) and $S_0(10)=0.5;$ and a Weibull baseline with an increasing hazard ($\kappa=1.5$) and $S_0(10)=0.5.$)
We also consider light censoring with $C \sim \min \left\{U(8, 12), \mbox{Exp}(0.02)\right\}$ and heavy censoring with $C\sim \min \left\{U(8, 12), \mbox{Exp}(0.20)\right\}$.

The results are reported in Table~\ref{tab:sim-tau10}. RAH-based test is more efficient than D-RMST-based test under light censoring, whereas D-RMST-based test is substantially more efficient than RAH-based test under heavy censoring.

\begin{table}[ht]
\centering
\caption{ARE and ERE under PH (HR$=0.70$), $\tau=10$, $n_1=n_0=150$ per arm and $B=500$.}
\label{tab:sim-tau10}
\resizebox{\textwidth}{!}{%
\begin{tabular}{lccccccc}
\toprule
Censoring & Baseline & ARE(D/C) & ARE(R/C) & ARE(R/D) & ERE(D/C) & ERE(R/C) & ERE(R/D)
\\
\midrule
Light & Weibull(kappa=0.7) & 0.920 & 0.981 & 1.067 & 0.985 & 0.974 & 0.989 \\
Light & Exponential(lambda0=0.10) & 0.913 & 0.976 & 1.068 & 0.941 & 0.989 & 1.051 \\
Light & Weibull(kappa=1.5) & 0.853 & 0.954 & 1.119 & 0.895 & 0.906 & 1.013 \\
\midrule
Heavy & Weibull(kappa=0.7) & 0.970 & 0.669 & 0.690 & 1.018 & 0.722 & 0.709 \\
Heavy & Exponential(lambda0=0.10) & 0.973 & 0.665 & 0.684 & 1.014 & 0.670 & 0.660 \\
Heavy & Weibull(kappa=1.5) & 0.987 & 0.651 & 0.659 & 0.987 & 0.611 & 0.619 \\
\bottomrule
\end{tabular}
}
\end{table}

\section{Simulation Study with Crossing Hazards}
\label{sec:appendix-crossing}

In this simulation, $\alpha(t)=\log(0.46)+(t/\tau_{\max})$ with $\tau_{\max}=43$ months,
$\kappa=0.826$, and $S_0(40)=0.50$: the HR starts at $0.46,$ crosses one, and rises above one late in follow-up. Thus, an early treatment benefit is reversed to harm at late times on the hazard scale. %
Table~\ref{tab:sim-crossing} reports ARE and ERE under light and heavy censoring. D-RMST-based test is more efficient than Cox HR-based test, whereas RAH-based test is less efficient than both Cox HR- and D-RMST-based tests, under both censoring regimes. The difference in efficiency is larger under heavy censoring.

This setting contrasts sharply with the alternatives from CheckMate~057 in Section~\ref{sec:cm057-analysis}.  There the HR also crosses one, but from above one and at month 2, so the treatment benefit dominates the entire analysis period and the RAH-based test is the most efficient. Here the crossing occurs near month 33 of the 40-month time window, and the higher hazard in the treatment group at later time points penalizes the tests that retain more weight near $\tau$.  The two cases illustrate that ``crossing hazards'' does not by itself determine the efficiency ordering; what matters is where the signal mass lies within the time window.

\begin{table}[ht]
\centering
\caption{ARE and ERE under crossing hazards; $\tau=40$, $S_0(40)=0.50$, $n_1=n_0=400$ per arm and $B=2000$.}
\label{tab:sim-crossing}
\begin{tabular}{lcccccc}
\toprule
Cens. & ARE(D/C) & ARE(R/C) & ARE(R/D) & ERE(D/C) & ERE(R/C) & ERE(R/D)
\\
\midrule
Light & 1.193 & 0.683 & 0.573 & 1.273 & 0.597 & 0.469 \\
\midrule
Heavy & 1.016 & 0.429 & 0.422 & 1.060 & 0.374 & 0.353 \\
\bottomrule
\end{tabular}
\end{table}

\bibliographystyle{plainnat}
\bibliography{refs}
\end{document}